\documentclass[12pt]{article}

\usepackage{fullpage}
\usepackage[T1]{fontenc}
\usepackage{lmodern}
\usepackage{graphicx}
\usepackage{amssymb}
\usepackage{amsmath}
\usepackage{subcaption}
\usepackage[round,authoryear]{natbib}
\usepackage{color}
\usepackage{array}
\usepackage{rotating}
\usepackage{colortbl}
\usepackage{tikz}
\usetikzlibrary{patterns}
\usepackage{enumitem}
\usepackage{hyperref}
\usepackage{upgreek}
\usepackage{mathrsfs}
\usepackage{cleveref}

\definecolor{webgreen}{rgb}{0,0.4,0}
\definecolor{webbrown}{rgb}{0.6,0,0}
\definecolor{purple}{rgb}{0.5,0,0.25}
\definecolor{darkblue}{rgb}{0,0,0.7}
\definecolor{darkred}{rgb}{0.7,0,0}
\definecolor{darkgreen}{rgb}{0,0.7,0}
\hypersetup{colorlinks,citecolor=darkred,filecolor=black,linkcolor=darkblue,urlcolor=webgreen}

\newcommand{\ignore}[1]{}
\newtheorem{lemma}{{\sc Lemma}}
\newtheorem{remark}{{\sc Remark}}
\newtheorem{prop}{{\sc Proposition}}

\newtheorem{theorem}{{\sc Theorem}}
\newtheorem{defn}{{\sc Definition}}

\newtheorem{example}{{\sc Example}}

\newenvironment{proof}{\noindent {\bf \sl Proof\/}:\enspace}
{\hfill $\blacksquare{}$ \vspace{12pt}}
\usepackage{sectsty}
\sectionfont{\centering\normalfont\scshape}
\subsectionfont{\centering\normalfont}

\Crefname{prop}{Proposition}{Propositions}

\title{{\Large {\bf Screening with advertisements}}
\thanks{I thank Bhavook Bharadwaj and Debasis Mishra for helpful comments and discussions.
}
}

\author{{\small Kolagani Paramahamsa
}~\thanks{
Independent Researcher; Email: \texttt{kolagani.paramahamsa@gmail.com}}}

\date{\small{\today}}

\begin{document}
\pagenumbering{roman}
\maketitle

\begin{abstract}
We investigate a seller’s revenue-maximizing mechanism in a setting where a desirable \emph{good} is sold together with an undesirable \emph{bad} (e.g., advertisements) that generates third-party revenue.
The buyer’s private information is two-dimensional: valuation for the good and willingness to pay to avoid the bad.
Building on the duality framework of \cite{Dask17}, whose results extend to our setting, we formulate the seller’s problem using a transformed measure~$\mu$ that depends on the third-party payment~$k$.
We provide a near-characterization for optimality of three pricing mechanisms commonly used in practice---the Good-Only, Ad-Tiered, and Single-Bundle Posted Prices---developing tractable, interpretable two-dimensional orthant conditions on~$\mu$ for sufficiency.
Economically, $k$ yields a clean comparative static: low~$k$ excludes the bad, intermediate~$k$ separates ad-tolerant and ad-averse buyers, and high~$k$ bundles ads for all types.

\vspace{10em}


\end{abstract}

\newpage

\pagenumbering{arabic}


\section{Introduction}

Subscription and ad‐supported digital services have become central to the global media economy. 
Over‐the‐top (OTT) video alone generated about \$169 billion in combined subscription and advertising revenue in 2024.\footnote{\href{https://www.pwc.com/id/gemo-2025}{PwC, Perspectives from the Global Entertainment \& Media Outlook 2025--2029}.}  
Platforms now offer \emph{rich and varied pricing menus}, ranging from purely ad‐supported access to premium ad‐free subscriptions, and increasingly hybrid forms that combine both. 
Netflix, for instance, introduced an ad-supported tier that now reaches more than 250 million monthly active viewers globally,\footnote{\href{https://about.netflix.com/en/news/netflix-upfront-2026-get-closer}{Netflix Upfront 2026: Get Closer, May 2026}.} while recent data show substantial consumer demand for both ad-supported and ad-free pricing options.\footnote{\href{https://www.antenna.live/insights/the-mix-and-match-era-how-streaming-consumers-really-pick-plans}{The Mix-And-Match Era: How Streaming Consumers Really Pick Plans, June 2026}.}
These developments highlight the economic relevance of hybrid monetization: combining subscriptions with advertising expands potential revenue but complicates incentive‐compatible pricing across heterogeneous users.

User heterogeneity arises from two sources---valuation of the service (the \emph{good}) and disutility from advertisements (the \emph{bad}) \citep{FerrerRichter2025}.
At its core, this environment therefore poses a two‐dimensional screening problem. 
The platform earns revenue not only from user payments but also from the advertiser whenever ads are shown. 
This additional revenue channel fundamentally alters the platform's incentives relative to standard multi‐good models, where all type dimensions represent desirable goods.
Here, one dimension represents a bad, and allocating it generates external compensation. 
The resulting trade‐off between extracting surplus from ad‐averse users and monetizing ad‐tolerant ones reshapes the geometry of incentive constraints---and hence the structure of optimal revenue-maximizing mechanisms.
The following stylized example illustrates these ideas.

\begin{example}
\upshape
Consider two equally likely users of types $x=(0.5,-0.2)$ and $x'=(1,-0.6)$, 
where the first coordinate denotes willingness to pay for the service and the magnitude of the second denotes willingness to pay to \emph{avoid} advertisements.  
Utility is additively separable, so a user of type $(x_1,x_2)$ consuming $(a_1,a_2)$---the good and the bad, respectively---and paying $p$ derives utility $x_1 a_1 + x_2 a_2 - p$.  
A third party (advertiser) pays the platform $0.1$ whenever ads are shown.
We compare three pricing menus commonly used by platforms.

\emph{Ad‐Tiered Posted Price.}\;
Suppose the platform offers two options: an ad-supported plan with $(a_1,a_2)=(1,1)$ priced at $0.3$, and an ad-free plan with $(a_1,a_2)=(1,0)$ priced at $0.9$.
User~$x$ (ad‐tolerant) chooses the ad‐supported plan since
$0.5(1) - 0.2(1) - 0.3 \ge \max\{0.5(1) - 0.2(0) - 0.9, 0\}$, 
while user~$x'$ (ad‐averse) chooses the ad‐free plan since
$1(1) - 0.6(0) - 0.9 \ge \max\{1(1) - 0.6(1) - 0.3, 0\}$.  
The platform’s expected revenue is therefore
\begin{equation*}
0.5(0.3+0.1) + 0.5(0.9) = 0.65.
\end{equation*}

\emph{Good‐Only Posted Price.}\;
If only the ad‐free plan at price~$0.5$ were available, both users would subscribe, 
yielding expected revenue $0.5$.

\emph{Single‐Bundle Posted Price.}\;
If only the ad‐supported plan at price~$0.3$ were offered, both users would subscribe, 
generating revenue $0.3+0.1=0.4$.

Hence, for a third‐party payment of $0.1$, the ad‐tiered menu generates higher than either alternative while satisfying incentive compatibility.\footnote{All prices shown are revenue‐maximizing for their respective menus.} 
A similar calculation shows that the Single-Bundle Posted Price yields higher revenue for $k>0.6$.
\end{example}

Our analysis formalizes this environment within a general multidimensional mechanism design framework.
We model the platform’s revenue–maximization problem using the duality approach of \cite{Dask17} (henceforth DDT), originally developed for multi–good settings.
In their formulation, the seller’s objective (expected revenue) takes the form $\int u,d\mu$, where $u$ is the buyer’s indirect utility and $\mu$ is a transformed measure derived from the buyer’s type density.
Introducing third–party payments changes the geometry of this measure: $\mu$ now also depends on the third–party payment $k$, which shifts its mass across the type space.
Importantly, the incentive compatibility and feasibility restrictions on $u$ remain unchanged: $u$ is convex, and its derivatives $u’$ (which correspond to allocations) lie in $[0,1]$.
Despite the change in the objective measure and the presence of a bad dimension capturing disutility from ads, DDT’s two central results, strong duality and the finite–menu characterization, continue to hold.

Building on that foundation, we develop a near-characterization of optimality for the three pricing mechanisms described above.
For each mechanism, we derive necessary conditions from the full DDT characterization and show that, under a mild assumption (a closely related condition first appeared in \cite{mcafee_mcmillan_1988} in a similar context), the remaining sufficiency conditions can be reduced to simple orthant tests on $\mu$.
The necessary conditions already impose economically meaningful restrictions on the level of third-party payment $k$.
In particular, the Good–Only Posted Price can be optimal only when $k$ is sufficiently low, while the Single–Bundle Posted Price requires $k$ to be sufficiently high.
The Ad–Tiered Posted Price can arise over a broader range of intermediate values of $k$, reflecting its ability to accommodate both ad-averse and ad-tolerant buyers.

The sufficiency conditions provide a common geometric interpretation of all three mechanisms.
For each type in an allocation region, we consider the corresponding orthant in the direction associated with that allocation. 
Our sufficiency conditions require every such orthant to have positive $\mu$-mass.
Intuitively, positive $\mu$-mass in an orthant favors increasing the derivative of the indirect utility in that direction.
Since these derivatives correspond to allocations, the direction of the orthant determines which allocation is favored.
Thus, northeast orthant positivity, corresponding to buyers with high valuation for the good and high tolerance for advertising, supports the maximal allocation $(1,1)$.
Southeast orthant positivity, in contrast, corresponds to buyers with high valuation for the good but greater aversion to advertising, and supports $(1,0)$ by favoring the maximal good allocation and the minimal bad allocation.
The orthant conditions therefore translate the directional geometry of $\mu$ into economically meaningful allocation patterns across the type space.
At intermediate levels of third-party payment, they can support both allocation patterns in different regions, allowing high-value buyers to separate according to their willingness to pay to avoid advertising.

We then apply these conditions to a class of uniform distributions, which provides a particularly transparent setting for the comparative statics.
Fixing any distribution in this class, we vary only the third-party payment $k$.
We show how the optimal mechanism changes as advertising becomes more profitable: low values of $k$ favor exclusion of the bad, intermediate values can support separation between ad-averse and ad-tolerant buyers, and high values favor bundling the good with advertising.

\paragraph{Contribution.}
Our results differ qualitatively from standard multi–good models.
In those settings, the transformed measure $\mu$ assigns fixed positive and negative regions determined solely by the distribution of buyer types.
In our model, by contrast, as the third–party payment $k$ increases, the geometry of the transformed measure $\mu$ evolves in a transparent way: the sign pattern along the boundaries of the type space shifts, altering which incentive constraints bind.
This rich structure makes comparative statics with respect to the third–party payment tractable even in this multidimensional environment.
For a class of uniform distributions, we fix the distribution and vary only $k$.
This yields a unified transition from Good–Only to Ad–Tiered to Single–Bundle mechanisms, showing how changes in third–party monetization interact with type heterogeneity in determining the optimal allocation.

Methodologically, the paper contributes a new class of tractable and interpretable sufficiency checks for two–dimensional screening.
The full DDT characterization provides necessary and sufficient conditions, but its sufficiency part requires verifying integral inequalities against all convex monotone functions.
For each of the three benchmark mechanisms, we derive necessary conditions from this full characterization and show that, under a mild additional condition, the remaining sufficiency requirement can be reduced to our orthant tests.
The Good–Only case also makes transparent why this reduction does not extend to a full characterization in general.
There, the structure of $\mu$ effectively reduces the convexity requirement in one allocation region to one dimension, allowing the sufficiency check to be carried out using a small generating family of convex functions.
In the other allocation region, convexity is genuinely two–dimensional, and no analogous small generating family is available.
Our orthant conditions provide a tractable two–dimensional way to handle this remaining requirement and yield interpretable sufficient conditions for optimality.

\subsection{Related Literature}

Our paper contributes to the literature on multidimensional screening.
Foundational work such as \cite{Rochet_Chone_1998} and \cite{Armstrong_1996} highlighted the analytical difficulties of characterizing optimal mechanisms when types are multidimensional.
Subsequent research (\cite{Manelli_Vincent_2006,Manelli_Vincent_2007, pavlov_2011}) has focused on identifying sufficient conditions under which simple mechanisms are optimal despite this complexity.
Our analysis follows this line of inquiry, extending it to environments that include a “bad’’ dimension and third–party compensation—features that lead to economically interpretable results.

A central reference for our analysis is the finite-menu characterization of DDT, which provides a general but computationally demanding test: optimality requires verifying an integral inequality for all convex–monotone functions.  
DDT also propose tractable sufficiency checks for the two-dimensional screening model that impose conditions on $\mu$ {\it separately} along each dimension.
We instead develop two-dimensional orthant conditions on $\mu$ that remain tractable and straightforward to verify.

Our results also connect to the pure–bundling literature, as our \emph{Single–Bundle Posted Price} is a pure-bundling mechanism in which the bundle consists of a desirable good and an undesirable bad.
 \cite{Haghpanah_Hartline_2021} provide conditions for pure bundling to be optimal in multi–good settings; \cite{Menicucci_Hurkens_Jeon_2015} derive sufficient conditions using the DDT approach, under stronger assumptions (e.g., independence / one–dimensional restrictions). 
 Our orthant–based sufficiency conditions are complementary: they are straightforward to verify for closed–form densities, and they naturally adapt to environments with a ``bad’’ dimension and third-party monetization.
More recently, \cite{Frick_Iijima_Ishii_2026} offer a different perspective on pure bundling: they show that as the seller’s information about buyers’ valuations becomes sufficiently precise, pure bundling achieves the same convergence rate to first-best revenue as the optimal mechanism.

\citet{Yang2026} studies multidimensional screening in which, in addition to screening a buyer based on a productive component, the principal can induce the buyer to take a costly action.
They characterize conditions under which the optimal mechanism screens only the productive component and does not use the costly action. 
This is related to our \emph{Good–Only Posted Price}: when the third–party payment is sufficiently low, our optimal mechanism does not allocate the bad despite heterogeneity in buyers’ willingness to pay to avoid it. 
The key difference is that in our setting, allocating the bad generates third–party revenue, so the platform’s decision to allocate it reflects a trade–off between the disutility it imposes on buyers and the revenue it generates.

Most closely related to our paper, \citet{Ichihashi_Jeon_Kim_2026} study mechanism design for ad-supported platforms in which advertising both generates revenue and imposes a nuisance cost on users. 
The key distinction from our setting is the dimensionality of private information. 
Users in their model have a one-dimensional type that simultaneously determines their preferences for the service and for advertising. 
As a consequence, users who value the service more highly are also more ad-averse, and their analysis exploits this one-dimensional type structure together with a rich, multidimensional allocation space and external advertising revenue. 
This structure allows them to obtain tractable results that can be applied to questions about how content quality changes with advertising policy. 
Our focus is different: the two-dimensional nature of users’ preferences is central to our analysis, as we study how distinct valuations for the good and willingness to pay to avoid the bad shape the optimal menu, and in particular when they give rise to the practical Good–Only, Ad–Tiered, and Single–Bundle Posted Prices considered in our paper. 
In \Cref{sec:oneD}, we present their model in our notation and discuss the differences between the two settings in greater detail.

\citet{GuttMehtaQuinn} study revenue-maximizing subscription mechanisms in which the platform chooses advertising intensity and subscription fees, while users endogenously choose their consumption of the service; advertising revenue therefore depends on both ad intensity and endogenous usage. Their baseline model has one-dimensional private information, while their extension allows users to be privately heterogeneous in both content valuation and ad nuisance. For this two-dimensional problem, they use a discretization-and-discounting approach to construct finite posted menus with revenue-loss bounds relative to the optimal mechanism. We instead treat the good and advertising as contract attributes and abstract from the user's subsequent consumption choices. This allows us to characterize directly how two-dimensional preferences over the good and the bad map into optimal allocations.

 Several other papers examine ad-tiered pricing under more restrictive assumptions.
\citet{Prasad2003} and \citet{Kumar2009} study related models in which the seller offers ad-supported and premium access, but reduce heterogeneity to a single-dimensional index or assume predetermined pricing structures.
\citet{Sato_2019} likewise analyzes ad-tiered pricing in a one-dimensional framework where heterogeneity is only along the bad dimension, and shows that an ad-tiered menu can be optimal.

Finally, a related strand of research studies advertising markets and media platforms as two-sided markets, emphasizing advertiser-side equilibrium.
For example, \citet{AndersonCoate2005} and \citet{PeitzValletti2008} study advertising intensity when ads impose a nuisance cost on consumers, while \citet{AndersonForosKind2018} and \citet{GentzkowShapiroYangYurukoglu2024} examine advertising in platform business models and advertiser competition, respectively.
This literature primarily focuses on advertiser-side equilibrium and welfare implications, whereas the present paper studies the user-side mechanism-design problem, treating ads as a contract attribute and characterizing optimal screening over service provision and advertising.

\paragraph{}

\noindent
The remainder of the paper is organized as follows.
\Cref{sec:model} introduces the model and establishes the mapping of DDT's results to our setting.
\Cref{sec:results} presents the main results—the necessary conditions and tractable sufficiency conditions for the three mechanisms—and provides a unified economic and geometric interpretation of how the optimal mechanism evolves as the third-party payment varies.
\Cref{sec:discussion} discusses extensions and broader implications, including the relaxation of the constant third-party payment assumption.
\Cref{sec:conclusion} concludes.


\section{A Screening Model with Ads}\label{sec:model}
A seller offers one unit each of a good and a bad to a buyer. 
The buyer’s type $x=(x_1,x_2)$ captures the valuations for the good and the bad in $x_1$ and $x_2$, respectively. 
Given an allocation $a=(a_1,a_2)$ and payment $p$ to the seller, a buyer of type $x$ derives utility $x \cdot a - p$, where $a_1$ and $a_2$ correspond to the good and the bad allocations, respectively. 
The second dimension represents a disutility, so we restrict $x_2 \le 0$; its magnitude $|x_2|$ is interpreted the buyer’s maximum willingness to pay to avoid the bad.

The valuation vector $x$ is the buyer's private information. 
The seller's belief is that the type $x$ is drawn from $X = [\underline{x}_1,\overline{x}_1] \times [\underline{x}_2,\overline{x}_2] \subset \mathbb{R}_{+} \times \mathbb{R}_{-}$ with density function $f: X \to \mathbb{R}_{+}$ that is continuous and differentiable with bounded derivatives.
Consistent with the interpretation of $x_2$, the second coordinate of the type space $X$ is restricted to $\mathbb{R}_{-}$.

The seller commits to a mechanism in which the buyer participates.
This mechanism generates revenue from two distinct sources. 
First, from the buyer’s direct payment under the mechanism. 
Second, from a third party who pays an amount $\kappa(x)$ to the seller for allocating the bad to a buyer of type $x$. 
The third--party payment function $\kappa: X \to \mathbb{R}_{+}$ is fixed a priori before the seller proposes a mechanism to the buyer. 
We assume that $\kappa$ is bounded, ruling out trivial cases.

In analyzing the seller's revenue-maximizing mechanism, we focus on direct mechanisms without loss of generality. 
A direct mechanism is a pair of functions $(q,t)$, where $q:X \to [0,1]^2$ is the allocation function and $t:X \to \mathbb{R}$ is the payment function. 
The first and second components of the allocation function $q$ represent the probability of allocations of the good and the bad, respectively.\footnote{Alternatively, these could be interpreted as quantities of the good and the bad.}
Hence, the expected utility of a buyer of type $x$ reporting $x’$ in mechanism $(q,t)$ is $x \cdot q(x’) - t(x’)$.

\begin{defn}[Incentive Compatibility]
A mechanism $(q,t)$ is incentive compatible (IC) if and only if 
\begin{equation*}
x \cdot q(x) - t(x) \;\geq\; x \cdot q(x') - t(x') \quad \forall \, x, x' \in X.
\end{equation*}
\end{defn}

\begin{defn}[Individual Rationality]
A mechanism $(q,t)$ is individually rational (IR) if and only if 
\begin{equation*}
x \cdot q(x) - t(x) \;\geq\; 0 \quad \forall \, x \in X.
\end{equation*}
\end{defn}

The utility function $u: X \to \mathbb{R}$ defined by $u(x) := x \cdot q(x) - t(x)$ represents the utility of a buyer of type $x$ under truthful reporting in the mechanism $(q,t)$.
It follows from \cite{rochet_1987} that $(q,t)$ is IC if and only if $u$ is convex and $q(x) = \nabla u(x)$ almost everywhere.

The seller’s revenue from type $x$ equals the buyer’s payment $t(x)$ plus the third--party payment $v_{\kappa(x)} \cdot q(x)$, where $v_{\kappa(x)} := (0,\kappa(x)).$ 
Since $t(x)=x\cdot q(x)-u(x)$ and $q(x)=\nabla u(x)$ a.e., revenue from type $x$ is
\begin{equation*}
(x + v_{\kappa(x)}) \cdot \nabla u(x) - u(x).
\end{equation*}

Feasibility of the allocation function implies $\nabla u(x) \in [0,1]^2$ a.e., while IR implies $u \geq 0$. 
Following DDT, we let $\mathcal{U}(X)$ be the set of convex, nondecreasing, continuous functions $u:X \to \mathbb{R}$ and $\mathcal{L}_1(X)$ denote the set of functions $u:X \to \mathbb{R}$ that are 1-Lipschitz with respect to the $\ell_1$ norm. 
We can then formulate the seller’s problem as:
\begin{center}
\begin{minipage}{0.6 \textwidth}
    \fbox{%
    \begin{minipage}{\dimexpr\linewidth-2\fboxsep-2\fboxrule}
    \vspace{-1pt}
        \begin{align}
        		\sup  & \int_{X} \big[(x + v_{\kappa(x)}) \cdot \nabla u(x) - u(x)\big] f(x) dx  \nonumber \\
                 \text{s. t.} \qquad & u \in \mathcal{U}(X) \cap \mathcal{L}_1(X) \nonumber\\
                & u(x) \geq 0, \forall x \in X \nonumber
        \end{align}
    \end{minipage}%
    }
\end{minipage}
\end{center}
We use the fact that the feasibility constraint $\nabla u(x) \in [0,1]^2$ is equivalent to requiring that $u$ is nondecreasing and 1-Lipschitz with respect to the $\ell_1$ norm. 
Convexity ensures continuity in the interior of $X$, which then extends to the boundary.

Denote $\underline{x} = (\underline{x}_1,\underline{x}_2)$. 
Following DDT, we normalize utilities by setting $u(\underline{x}) = 0$ without loss of generality.
Equivalently, the IR constraint can be omitted by writing the objective as
\begin{equation}\label{objIR}
\int_X \Big[(x + v_{\kappa(x)}) \cdot \nabla u(x) - \big(u(x) - u(\underline{x})\big)\Big] f(x)\,dx.
\end{equation}

We apply integration by parts (as in \cite{mcafee_mcmillan_1988}, with details in Appendix \ref{sec:intbyparts}) to the objective \eqref{objIR}, obtaining
\begin{equation}\label{objDiv}
\int_{\partial X} u(x) f(x) \big((x + v_{\kappa(x)}) \cdot \hat{n}\big) \, dx 
- \int_{X} u(x) \Big(\nabla f(x) \cdot (x + v_{\kappa(x)}) + (3 + \partial_2 \kappa(x)) f(x)\Big) \, dx 
\, + \, u(\underline{x}).
\end{equation}
Relative to the standard multi-good formulation, the difference is the presence of the additional terms involving $\kappa(x)$.

We define a transformed measure $\mu$ on $X$ by\footnote{Since the objective is linear in continuous and compactly supported $u$, by Riesz’s representation theorem we are defining a Radon measure.}
\begin{defn}
For all measurable sets $A \subseteq X$,  
\begin{equation*}
\mu(A) := \int_{\partial X} \mathbb{I}_{A}(x) f(x) \big((x + v_{\kappa(x)}) \cdot \hat{n}\big) \, dx 
\,-\, \int_{X} \mathbb{I}_{A}(x)\big(\nabla f(x) \cdot (x + v_{\kappa(x)}) + (3 + \partial_2 \kappa(x)) f(x)\big) \, dx  
\,+\, \mathbb{I}_{A}(\underline{x}).
\end{equation*}
\end{defn}

Note that $\mu(X)=0$. 
This follows by evaluating $\eqref{objIR}$ and $\eqref{objDiv}$ at $u(x)=1$.
This observation motivates the inclusion of the $u(\underline{x})$ term in the objective.
The normalization $\mu(X)=0$ will simplify the analysis.

Given the definition of $\mu$, the optimal program is
\begin{align}
        	\sup  & \int_{X} u(x) \, d\mu  \nonumber \\
         \text{s.t.} \qquad & u \in \mathcal{U}(X) \cap \mathcal{L}_1(X). \nonumber
\end{align}

Our formulation differs from that of DDT in two respects.
First, the second component of the type space $X$ is supported on $\mathbb{R}_-$. 
Second, the transformed measure includes additional terms involving $\kappa(x)$. 
For the main results in DDT to carry over, these differences are immaterial, since the key features driving the proofs remain unchanged.

\subsection{Characterization Result of DDT}

DDT establish a strong duality between the seller's problem and a transport formulation over $X\times X$.  
Because $X\subseteq\mathbb{R}^n$ is convex and compact, and the additional $\kappa(x)$ terms enters linearly in $u$, their arguments extend to our setting.  
Hence, the strong-duality result continues to hold in our framework.\footnote{\cite{Kleiner2019} provide an alternative linear-programming proof of strong duality and note that their argument extends to arbitrary signed Radon measures~$\mu$.}

Building on this dual formulation, they characterize the optimality of finite-menu mechanisms using integral inequalities involving \emph{convex monotone test functions}.  
The same characterization applies directly in our setting.  
To state it, we recall two pieces of notation:
\begin{itemize}[itemsep=-2pt, topsep=1pt, partopsep=0pt]
\item For a signed measure $\mu$ and measurable $A\subseteq X$, the restriction of $\mu$ to $A$ is $\mu|_A(S)=\mu(A\cap S)$ for all measurable $S$.
\item The \emph{menu} of a mechanism $(q,t)$ is $\text{Menu}_{(q,t)}=\{(a,p):\exists\,x\in X\text{ such that }(a,p)=(q(x),t(x))\}$.
\end{itemize}

\begin{defn}
For $a\in[0,1]^2$, define $\vec v(a)\in\{-1,0,1\}^2$ by 
$v_i=1$ if $a_i=0$, $v_i=-1$ if $a_i=1$, and $v_i = 0$ if $a_i \in (0,1)$. 
A function $u:X\to\mathbb{R}$ is \emph{convex $\vec v(a)$-monotone} if it is convex, nondecreasing in each coordinate with $v_i=1$, and nonincreasing in each coordinate with $v_i=-1$.
\end{defn}

\begin{lemma}[DDT]\label{lem:finitechar}
Let $\mu$ be the transformed measure.  
A mechanism $(q,t)$ with a finite menu is optimal if and only if, for every $(a,p)\in\text{Menu}_{(q,t)}$,
\begin{equation*}
\int u\,d \mu |_{R} \;\leq\; 0
\end{equation*}
for all convex $\vec v(a)$-monotone functions $u$, 
where $R=\{x\in X:(q(x),t(x))=(a,p)\}$ is the set of types receiving $(a,p)$.
\end{lemma}

The convexity of the test functions reflects primal feasibility, while the complementary-slackness conditions in DDT's dual program identify the transport directions that bind those constraints—thereby inducing the coordinatewise monotonicity requirements in the test functions.


\section{Results}\label{sec:results}
We extend the analysis of DDT to the good–bad bundle environment.
To this end, we restrict attention to the case of a constant third--party payment,
\begin{equation*}
\kappa(x) = k \quad \text{for some } k \in \mathbb{R}_+,
\end{equation*}
so that the seller receives a fixed payment $k$ whenever the bad is allocated, independent of the buyer’s type.
This assumption reflects common practice in subscription markets, where advertiser payments depend on impressions rather than user characteristics.
While our arguments extend to more general $\kappa(x)$ (see \Cref{dis:k}), we adopt the constant case for clarity of exposition.

Under this restriction, let $v_k := (0,k)$. 
For all measurable sets $A \subseteq X$, the transformed measure $\mu$ is defined by
\begin{equation}\label{eq:transk}
\mu(A)
= \int_{\partial X} \mathbb{I}_{A}(x)\, f(x)\, \big((x + v_k)\!\cdot\!\hat{n}\big)\,dx
- \int_{X} \mathbb{I}_{A}(x)\, \big(\nabla f(x)\!\cdot\!(x + v_k) + 3f(x)\big)\,dx
+ \mathbb{I}_{A}(\underline{x}).
\end{equation}

\noindent {\bf Geometry of $\mu$ and dependence on $k$.}
As the third-party payment $k$ increases, the boundary component of $\mu$, given by $f(x)(x+v_k)\!\cdot\!\hat{n}$, changes sign along the top and bottom edges of $X$, which represent the disutility dimension $x_2$.
For small $k$, $\mu$ is positive at the bottom (ad-averse types) and negative at the top.
As $k$ rises, the top edge becomes positive, and for sufficiently large $k$, the bottom edge turns negative.
These sign reversals anticipate the three canonical mechanisms studied below and set up the detailed economic interpretation provided at the end of this section.

For our sufficiency results, we impose the following condition, first introduced by \cite{mcafee_mcmillan_1988}.
\begin{defn}\label{def:cond}
A density function $f$ satisfies \emph{\bf MM (McAfee--McMillan)} if
\begin{equation*}
\nabla f(x) \cdot (x + v_k) + 3 f(x) \;\geq\; 0 \quad \text{for all } x \in X.
\end{equation*}
\end{defn}

This condition ensures that the induced measure $\mu$ assigns nonpositive mass throughout the interior of~$X$, similar to assumptions commonly imposed in the multi–good mechanism design literature (\citealp{pavlov_2011, Menicucci_Hurkens_Jeon_2015, deb_bik}).
While the classical model imposes $\nabla f(x) \cdot x + 3f(x)\!\ge\!0$, our version adapts it to incorporate third–party revenue, becoming slightly stronger when $\partial_2 f(x)<0$ and weaker otherwise (see \Cref{dis:regularity} for illustration).

We now analyze three practical mechanisms that arise naturally in subscription markets:
(i) selling the good alone,
(ii) selling only a bundled product that includes the bad, and
(iii) offering both tiers simultaneously, allowing buyers to pay to avoid the bad.
We then specialize to a uniform distribution, which allows us to connect these results and characterize how the optimal mechanism changes with the third-party payment.


\subsection{Good-Only Posted Price}
Services such as Apple Music or Netflix (in markets without an ad-supported tier) are offered exclusively in an ad-free form at a posted price. 
The menu in this case consists of two outcomes: one that offers the good alone at a posted price, and the default outcome in which the buyer pays nothing and receives nothing.

\begin{defn}[Good-Only Posted Price]\label{def:good-only}
A mechanism $(q,t)$ is called a \emph{Good-Only Posted Price} mechanism if there exists a price $p_{g}^* \in [\underline{x}_1,\overline{x}_1]$ such that
\begin{equation*}
\big(q(x),t(x)\big) =
\begin{cases}
\big((0,0),\,0\big) & \text{if } x_1 \leq p_{g}^*,\\
\big((1,0),\,p_g^*\big) & \text{if } x_1 > p_{g}^*.
\end{cases}
\end{equation*}
\end{defn}

A mechanism like Good-Only Posted Price that screens only one of the dimensions is never optimal in a multi-good environment.
In the two-good case, its menu is $\{\big((0,0),0\big),\big((1,0),p_{g}^*\big)\}$.
Augmenting the menu with the outcome $\big((0,1),p_{g}^*\big)$ strictly increases revenue: types $\{x \in X : x_{1} \le p_{g}^*,\, x_{2} > p_{g}^*\}$—who previously chose the default $((0,0),0)$—now select the new option and pay $p_{g}^*$, while all other types’ choices remain unchanged.

By slight abuse of notation, for a Good-Only Posted Price mechanism with price $p_{g}^*$, let
\begin{equation*}
\mathcal{Z} := \{x \in X : x_{1} \leq p_{g}^*\} \quad \text{and} \quad \mathcal{W} := X \setminus \mathcal{Z}.
\end{equation*}
Types in $\mathcal{W}$ buy the good at price $p_g^*$, while those in $\mathcal{Z}$ choose the default option (\Cref{fig:good-only}).

\begin{figure}[!hbt]
\centering
\includegraphics[width=3in]{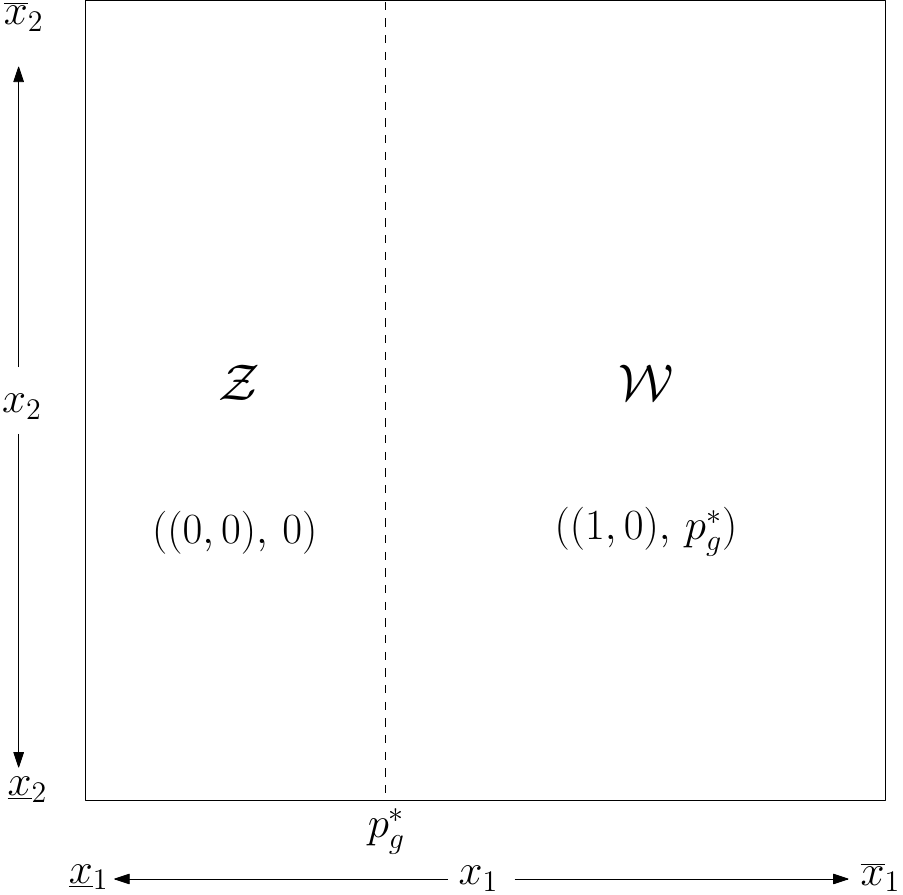}
\caption{Good-Only Posted Price}
\label{fig:good-only}
\end{figure}

In what follows, we first derive necessary conditions for the optimality of such mechanisms. 
For types in $\mathcal{Z}$, these conditions are also sufficient, whereas sufficiency for $\mathcal{W}$ further requires an orthant condition on the measure $\mu$.

\begin{prop}\label{prop:good-only}
If a \emph{Good-Only Posted Price} mechanism is optimal, then the following conditions must hold:
\begin{enumerate}[label=\textup{(\roman*)}]
    \item $\mu(\mathcal{Z}) = \mu(\mathcal{W}) = 0$;
    \item $k \leq |\overline{x}_2|$;
    \item For all $t \in [\underline{x}_1, p_{g}^*]$,
    \begin{equation*}
        \int_{t}^{p_{g}^*} M(x_1)\,dx_1 \;\geq\; 0,
        \qquad 
        \text{where } M(x_1) := \mu\big([\underline{x}_1, x_1] \times [\underline{x}_2, \overline{x}_2]\big).
    \end{equation*}
\end{enumerate}
\end{prop}

Condition (ii) requires that the third--party payment $k$ be no greater than the minimum disutility from the bad, $|\overline{x}_2|$.
Condition (iii) requires that tail integrals of the one-dimensional marginal cumulative measure $M(x_1)$ be nonnegative on $\mathcal{Z}$.

\noindent \emph{Proof Sketch.} 
Condition (i) follows directly from \Cref{lem:finitechar}, since convex dominance against constants forces each menu region to have zero total mass under $\mu$. 
For (ii), suppose $k > |\overline{x}_2|$. 
Then the top edge of $X$ carries strictly positive mass under $\mu$, contradicting the dominance condition for $\mathcal{Z}$ in \Cref{lem:finitechar} when tested against convex functions that increase steeply in $x_2$ near the top edge.
Finally, for (iii) we apply \Cref{lem:finitechar} with hinge test functions $u_t(x) = (x_1 - t)_+$ supported on $\mathcal{Z}$, which reduce the condition to the stated integral inequality. 
Details are provided in Appendix~\ref{appendix:good-only}.

The conditions in Proposition \ref{prop:good-only} are also sufficient for optimality in $\mathcal{Z}$. 
Adding one further condition extends sufficiency to $\mathcal{W}$.

\begin{theorem}\label{theo:good-only}
Suppose $f$ satisfies MM, and the \emph{Good-Only Posted Price} mechanism with price $p_g^*$ satisfies Conditions \textup{(i)}–\textup{(iii)} of Proposition~\ref{prop:good-only}. 
If, in addition,
\begin{enumerate}[label=(\roman*),resume]
\item[\textup{(iv)}] $\mu \big([x_1,\overline{x}_1]\times[\underline{x}_2,x_2]\big) \geq 0
\quad\text{for all }x=(x_1,x_2)\in\mathcal W$,
\end{enumerate}
then the mechanism is optimal.
\end{theorem}

\noindent \emph{Proof Sketch.}
The argument separates the two regions of the type space.
For $\mathcal{Z}$, by \Cref{lem:finitechar}, the test functions are convex and nondecreasing in both coordinates.
Under MM, Condition (i) of Proposition~\ref{prop:good-only} and the fact that the left edge carries negative $\mu$–mass, implies that positive mass in $\mathcal{Z}$ is supported only on the bottom edge.
It therefore suffices to test one-dimensional convex, nondecreasing functions supported on the bottom edge.
These are represented by hinge functions (Theorem 1.5.7 in \cite{muller2002comparison}), precisely those used to derive Condition (iii) in Proposition~\ref{prop:good-only}.

\noindent For $\mathcal{W}$, the test functions are convex, nonincreasing in $x_1$ and nondecreasing in $x_2$. 
Their sublevel sets are lower sets, which can be enclosed in lower–right orthants. 
Condition (i) and MM imply that enlarging to the orthant does not add positive mass, while Condition (iv) ensures every such orthant has nonnegative mass. 
Using a layer-cake approximation, we show that all admissible test functions integrate to nonpositive values with respect to $\mu|_{\mathcal{W}}$. Details are provided in Appendix~\ref{appendix:good-only2}.

Together, Proposition \ref{prop:good-only} and Theorem \ref{theo:good-only} provide a near-characterization of the Good-Only Posted Price mechanism: Conditions (i)–(iii) are necessary, and adding Condition (iv) makes them sufficient.

\begin{remark}\label{rem:convexity}
\upshape
In region $\mathcal{Z}$, the structure of $\mu$ restricts positive mass to the bottom edge, effectively reducing the analysis to one dimension.
In this setting, every convex nondecreasing function can be expressed as a countable superposition of hinge functions $(x-t)_+$, each corresponding to a supporting hyperplane where the slope changes.
Hence verifying the dominance inequality only for these one-dimensional generators (as in Condition (iii)) is both necessary and sufficient.
In higher dimensions, however, convexity can occur along uncountably many directions, and the cone of convex functions admits no finite or countable generating family.
As \citet{muller2002comparison} note, “in the multivariate case there is no easy characterization of these orders as in Theorem 1.5.7.… Therefore there is no hope of finding a ‘small’ generator.”
Thus, the reduction to hinge functions that works in $\mathcal{Z}$ is no longer possible once convexity becomes genuinely two-dimensional.
\end{remark}

\subsubsection{The case without third--party payments}

We next consider the benchmark case without third–party payments ($k=0$), where allocating the bad yields no external compensation.
The following two examples illustrate this case: the first shows that bundling can raise revenue even when $k=0$, and the second verifies the optimality of the \emph{Good–Only Posted Price} mechanism under the conditions of \Cref{theo:good-only}.

\begin{example}\label{ex:fullsurplus}
\upshape
Let $X = [0,1]\times[-1,0]$ and let $f$ place equal probability $0.5$ on the two points $x=(0.5,0)$ and $x'=(1,-1)$.\footnote{Although discrete, this distribution can be approximated by a limit of truncated normal densities centered at $x$ and $x’$.}
Consider the mechanism $(q,t)$ defined by $(q(x),t(x))=((1,1),0.5)$ and $(q(x'),t(x'))=((1,0),1)$. 
This mechanism is incentive compatible and yields expected revenue $0.75$, while the optimal Good-Only Posted Price (either $1$ or $0.5$) yields only $0.5$ revenue. 
Thus bundling the bad with the good allows the seller to extract strictly higher revenue.
\end{example}

\begin{example}
\upshape
Let $X=[0,1]\times[-1,0]$, $k=0$, and let $f$ be uniform. 
Then for all measurable $A \subseteq X$, the transformed measure is
\begin{equation*}
\mu(A) = \mathbb{I}_{A}((0,-1)) + \int_{\partial X} \mathbb{I}_{A}(x)\,(x\cdot\hat{n})\,dx - 3\!\int_X \mathbb{I}_{A}(x)\,dx.
\end{equation*}
Hence $\mu$ assigns a point mass of $+1$ at $(0,-1)$, total mass $-3$ uniformly over the interior, and linear density $+1$ along the edges $x_2=-1$ and $x_1=1$.
Consider the Good-Only Posted Price mechanism with price $p_{g}^* = 0.5$.

\noindent It is immediate that Condition \textup{(i)--(ii)} of Proposition~\ref{prop:good-only} are satisfied. 
To verify Condition \textup{(iii)} of Proposition \ref{prop:good-only} on $\mathcal Z$, note that $M(x_1)=1+x_1-3x_1=1-2x_1$. 
Hence $\int_t^{0.5}(1-2x_1)\,dx_1=0.25-(t-t^2)\ge0$ for all $t\in[0,0.5]$. 
For Condition \textup{(iv)} of \Cref{theo:good-only}, we compute
\begin{equation*}
\mu([x_1,1]\times[-1,x_2])=(1-x_1)+(x_2+1)-3(1-x_1)(x_2+1)
=2x_1-2x_2+3x_1x_2-1 \ge 0
\end{equation*}
for all $x\in\mathcal W$. 
Thus the \emph{Good--Only Posted Price} mechanism is optimal in this example.
\end{example}


\subsection{Single-Bundle Posted Price}
Legacy print and broadcast media services are often bundled with advertisements and offered at a posted price.
The menu features two outcomes: one offering the bundle of the good and the bad at a fixed price, and a default option where the buyer pays and receives nothing.
Buyers do not have the option to reduce advertisement exposure by paying more.

\begin{defn}[Single-Bundle Posted Price]\label{def:single-bundle}
A mechanism $(q,t)$ is called a \emph{Single-Bundle Posted Price} mechanism if there exists a price $p_{sb}^* \in \mathbb{R}_+$ such that
\begin{equation*}
\big(q(x),t(x)\big) =
\begin{cases}
\big((0,0),\,0\big) & \text{if } x_1 + x_2 \leq p_{sb}^*,\\[4 pt]
\big((1,1),\,p_{sb}^*\big) & \text{otherwise}.
\end{cases}
\end{equation*}
\end{defn}

Given a \emph{Single-Bundle Posted Price} mechanism $(q,t)$ with price $p_{sb}^*$, define
\begin{align*}
\mathcal{Z} &:= \{\,x \in X : x_1 + x_2 \leq p_{sb}^*\,\},\\
\mathcal{Y} &:= X \setminus \mathcal{Z}.
\end{align*}
Types in $\mathcal{Y}$ purchase the bundle, while those in $\mathcal{Z}$ choose the outside option (see \Cref{fig:sb}).

\begin{figure}[!hbt]
\centering
\includegraphics[width=3in]{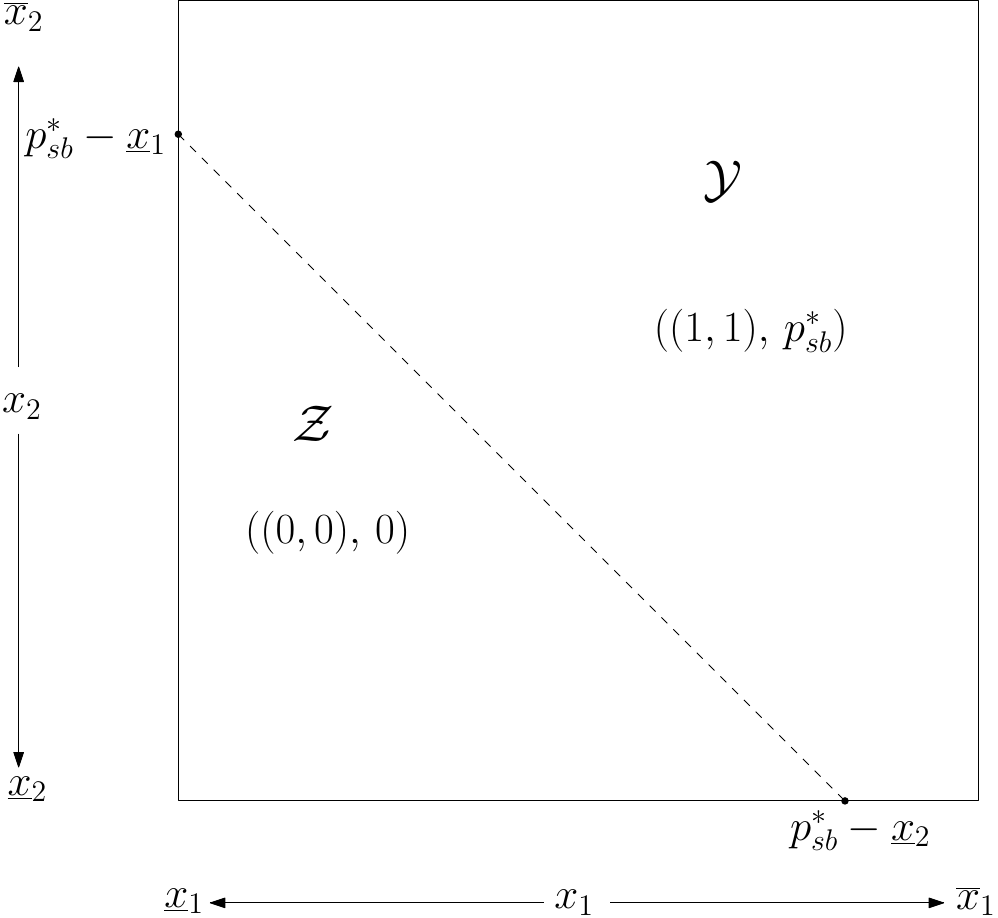}
\caption{Single-Bundle Posted Price}
\label{fig:sb}
\end{figure}

\begin{prop}\label{prop:sbnec}
If a \emph{Single-Bundle Posted Price} mechanism $(q,t)$ with price $p_{sb}^*$ is optimal, then the following conditions must hold:
\begin{enumerate}[label=\textup{(\roman*)}]
\item $\mu(\mathcal{Z}) = \mu(\mathcal{Y}) = 0$;
\item $k \geq |\underline{x}_2|,\;$ and $\;p_{sb}^* \leq \min\{\overline{x}_1 + \underline{x}_2,\;\underline{x}_1 + \overline{x}_2\}$.
\end{enumerate}
\end{prop}

Condition (ii) establishes necessary bounds on the third-party payment and bundle price for the \emph{Single-Bundle Posted Price} mechanism to be optimal.

\noindent\emph{Proof Sketch.}
Condition (i) follows immediately from Lemma~\ref{lem:finitechar}, as noted earlier in Proposition \ref{prop:good-only}.
Condition (ii) follows by testing the admissible convex–monotone functions from Lemma~\ref{lem:finitechar} against the sign structure of $\mu$ along the boundary of $X$.
When $k < |\underline{x}_2|$, positive mass on both the right and bottom edges makes it impossible to satisfy the dominance condition.
When $k \ge |\underline{x}_2|$, positivity on the top and right edges restricts the cutoff $p_{sb}^*$ to lie below $\min\{\overline{x}_1+\underline{x}_2,\;\underline{x}_1+\overline{x}_2\}$.
Details are given in \Cref{appendix:sbnec}.

Next, we show that combining conditions from Proposition \ref{prop:sbnec} with an orthant condition on $\mu$ is sufficient for the optimality of the \emph{Single-Bundle Posted Price} mechanism.

\begin{theorem}\label{theo:sbsuff}
Suppose $f$ satisfies MM, and the \emph{Single-Bundle Posted Price} mechanism with price $p_{sb}^*$ satisfies Conditions \textup{(i)--(ii)} of Proposition~\ref{prop:sbnec}. If, in addition, the following condition is satisfied:
\begin{enumerate}[label=(\roman*),resume]
\item[\textup{(iii)}] $\mu([x_1, \overline{x}_1] \times [x_2, \overline{x}_2] \cap \mathcal{Y} ) \geq 0
\quad\text{for all }x \in X$.
\end{enumerate}
then the mechanism is optimal.
\end{theorem}

\noindent\emph{Proof Sketch.}
The sufficiency conditions from Lemma~\ref{lem:finitechar} are verified separately on $\mathcal{Z}$ and $\mathcal{Y}$.
In $\mathcal{Z}$, Condition~(ii) of Proposition~\ref{prop:sbnec} and MM ensure that $\mu$ places positive mass only at the lowest point $(\underline{x}_1,\underline{x}_2)$, so all convex nondecreasing tests integrate to nonpositive values.
In $\mathcal{Y}$, the argument parallels Theorem~\ref{theo:good-only}: the orthant condition (iii) guarantees nonnegative mass on all upper-right orthants, implying that $\int u\,d\mu|_{\mathcal{Y}}\le0$ for all nonincreasing convex tests.
Details are in \Cref{appendix:sbsuff}.

Proposition~\ref{prop:sbnec} and Theorem~\ref{theo:sbsuff} jointly yield a near-characterization of the Single-Bundle Posted Price mechanism, with Conditions (i)–(ii) necessary and Condition (iii) completing sufficiency.
The following example demonstrates such an optimal mechanism.

\begin{example}\label{ex:sb}
\upshape
Let $X = [1,2] \times [-1,0]$, $\kappa(x) = 1.5$, and let $f$ be the uniform distribution. 
The transformed measure $\mu$ is
\begin{equation*}
\mu(A) = \mathbb{I}_{A}((1,-1)) 
+ \int_{\partial X} \mathbb{I}_{A}(x)\,\big((x_1, x_2+1.5)\cdot \hat{n}\big)\,dx 
- 3\int_{X} \mathbb{I}_{A}(x)\,dx, \qquad A \subseteq X.
\end{equation*}

\noindent Thus $\mu$ has: a point mass of $+1$ at $(1,-1)$; uniform interior mass $-3$ over $X$; linear densities $-0.5$ on $x_2=-1$, $+2$ on $x_1=2$, $+1.5$ on $x_2=0$, and $-1$ on $x_1=1$.

\noindent Let $p_{sb}^*$ be the positive root of $3p^2+3p-2=0$, i.e.\ $p_{sb}^*\approx 0.46$. 
We obtain $p_{sb}^*$ by solving $\mu(\mathcal{Z})=0$. 
Therefore, Condition \textup{(i)} of Proposition \ref{prop:sbnec} is satisfied by construction, Condition \textup{(ii)} is straightforward to verify. We check Condition \textup{(iii)} of \Cref{theo:sbsuff} in \Cref{excal:sb} and conclude that the \emph{Single-Bundle Posted Price} mechanism with price $p_{sb}^* \approx 0.46$ is optimal.
\end{example}


\subsection{Ad-Tiered Posted Price}\label{sec:ad-tiered}

Building on the two benchmark mechanisms, we introduce a more general design that allows the buyer to pay to avoid the bad while still obtaining the good.
This structure, characteristic of digital markets with ad-supported and premium tiers, is referred to as the \emph{Ad-Tiered Posted Price} mechanism.

\begin{defn}[Ad-Tiered Posted Price]\label{def:ad-tiered}
A mechanism $(q,t)$ is called an \emph{Ad-Tiered Posted Price} mechanism if there exist prices 
$p_{g}^* \in [\underline{x}_1,\overline{x}_1]$ and $p_{sb}^* \in \mathbb{R}_+$ satisfying 
$p_{sb}^* \in [p_{g}^* + \underline{x}_2, \min\{p_{g}^*,\underline{x}_1+\overline{x}_2\}]$ such that
\begin{equation*}
(q(x),t(x)) =
\begin{cases}
\big((0,0),\,0\big) & \text{if } x_1 \leq p_{g}^* \text{ and } x_2 \leq p_{sb}^* - x_1,\\[4pt]
\big((1,0),\,p_{g}^*\big) & \text{if } x_1 > p_{g}^* \text{ and } x_2 \leq p_{sb}^* - p_{g}^*,\\[4pt]
\big((1,1),\,p_{sb}^*\big) & \text{otherwise}.
\end{cases}
\end{equation*}
\end{defn}

For an Ad-Tiered Posted Price mechanism, the type space is partitioned according to the three menu items each type may choose (also see \Cref{fig:ad-tier}):
\begin{align*}
\mathcal{Z} &:= \{x \in X : x_1 \le p_{g}^*,\, x_2 \le p_{sb}^* - x_1\}, 
&& \text{(default option)},\\
\mathcal{W} &:= \{x \in X : x_1 > p_{g}^*,\, x_2 \le p_{sb}^* - p_{g}^*\}, 
&& \text{(ad-free service)},\\
\mathcal{Y} &:= X \setminus (\mathcal{Z} \cup \mathcal{W}),
&& \text{(service with advertisements)}.
\end{align*}

\begin{figure}[!hbt]
\centering
\includegraphics[width=3.5in]{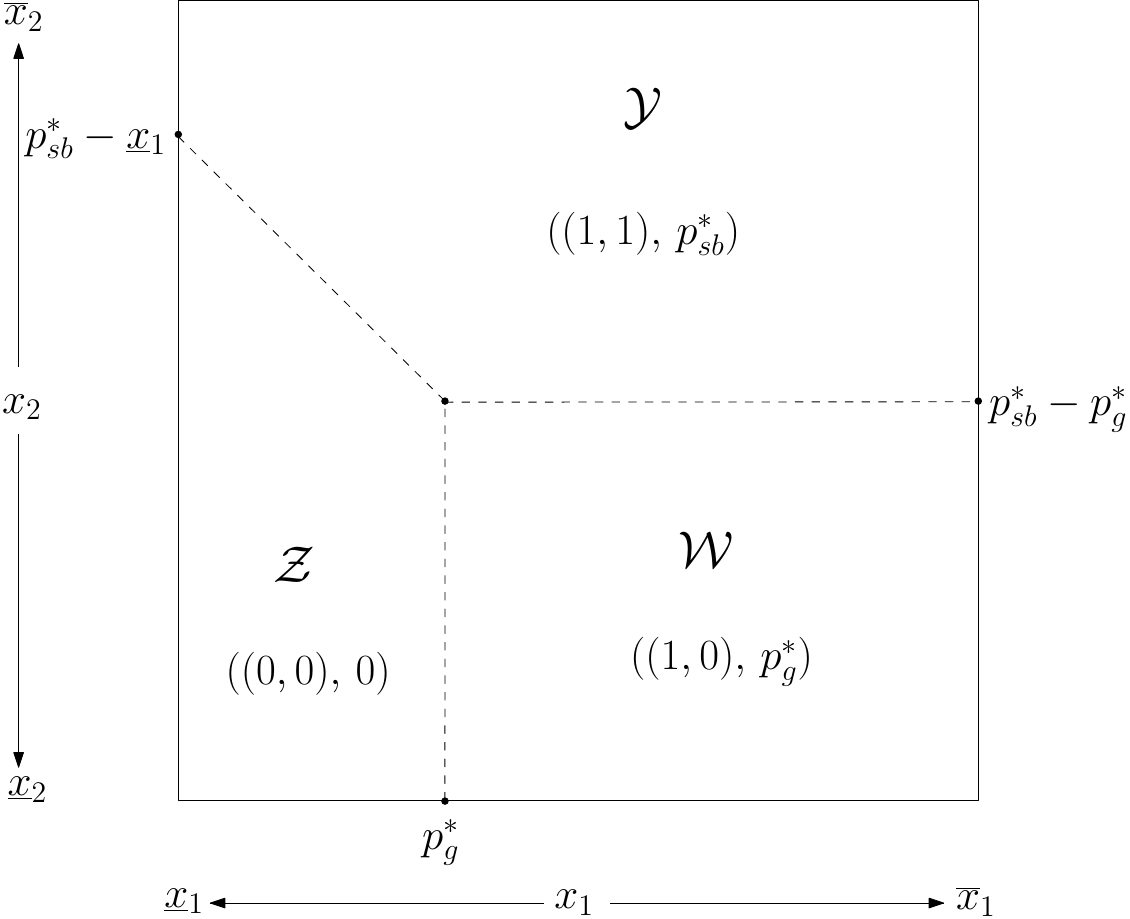}
\caption{Ad-Tiered Posted Price}
\label{fig:ad-tier}
\end{figure}

\begin{remark}\label{rem:pedagogy}
\upshape
The restriction $p_{sb}^* \le p_{g}^*$ follows from incentive compatibility, 
while the bound $p_{sb}^* \le \underline{x}_1 + \overline{x}_2$ defines the particular variant of the Ad-Tiered Posted Price mechanism analyzed here. 
The complementary case ($p_{sb}^* > \underline{x}_1 + \overline{x}_2$) necessitates additional conditions like $k \leq |\overline{x}_2|$ in \Cref{prop:good-only}.
However, the analysis and results essentially remain parallel, so we do not present it here.
The condition $p_{sb}^* \ge p_g^*+\underline{x}_2$ ensures that the indifference boundary $x_2=p_{sb}^*-p_g^*$ intersects the type space, so that the ad-supported region is nonempty.
\end{remark}

\begin{prop}\label{prop:ad-tier}
If an \emph{Ad-Tiered Posted Price} mechanism $(q,t)$ with prices $p_{g}^*$ and $p_{sb}^*$ is optimal, then the following conditions must hold:
\begin{enumerate}[label=\textup{(\roman*)}]
\item $\mu(\mathcal{Z}) = \mu(\mathcal{W}) = \mu(\mathcal{Y}) = 0$;
\item For all $t \in [\underline{x}_1, p_{g}^*]$,
    \begin{equation*}
        \int_{t}^{p_{g}^*} M(x_1)\,dx_1 \;\geq\; 0,
        \qquad 
        \text{where } M(x_1) := \mu\big(\big([\underline{x}_1, x_1] \times [\underline{x}_2, \overline{x}_2]\big) \cap \mathcal{Z} \big).
    \end{equation*}
\end{enumerate}
\end{prop}

Condition~(i) follows directly from \Cref{lem:finitechar}, while Condition~(ii) is derived through arguments analogous to those in Proposition~\ref{prop:good-only} for the region~$\mathcal{Z}$.

\begin{theorem}\label{theo:ad-tier}
Suppose $f$ satisfies MM, and the \emph{Ad-Tiered Posted Price} mechanism with prices $p_{g}^*$ and $p_{sb}^*$ satisfies Conditions~\textup{(i)--(ii)} of Proposition~\ref{prop:ad-tier}. 
If, in addition, the following conditions hold:
\begin{enumerate}[label=\textup{(\roman*)}]
\item[\textup{(iii)}] $\mu([x_1,\overline{x}_1] \times [\underline{x}_2,x_2] \cap \mathcal{W}) \geq 0
\quad\text{for all }x\in X$,
\item[\textup{(iv)}] $\mu([x_1,\overline{x}_1] \times [x_2,\overline{x}_2] \cap \mathcal{Y}) \geq 0
\quad\text{for all }x\in X$,
\end{enumerate}
then the \emph{Ad-Tiered Posted Price} mechanism $(q,t)$ with prices $p_{g}^*$ and $p_{sb}^*$ is optimal.
\end{theorem}

The sufficiency argument for $\mathcal{Z}$ parallels that in \Cref{theo:good-only}, 
relying on the nonpositivity of interior mass (MM) and the nonpositive density along the left edge, noting that there is no top edge to account for in this case. 
The arguments for $\mathcal{W}$ and $\mathcal{Y}$ follow those in \Cref{theo:good-only} and \Cref{theo:sbsuff}, respectively, with minor adjustments. 
For $\mathcal{W}$, the top edge carries no linear mass—unlike in \Cref{theo:good-only}, where the sign of the top-edge density played a crucial role.
For $\mathcal{Y}$, the absence of bottom-edge linear mass allows us to bypass the step in \Cref{theo:sbsuff} where the negative bottom edge played a key role. 
In both cases, Conditions~(iii) and~(iv) ensure that the sufficiency condition from \Cref{lem:finitechar} is satisfied.
The detailed proof is omitted for brevity.

\begin{remark}\label{rem:bounds}
\upshape
Propositions \ref{prop:good-only} and \ref{prop:sbnec} establish that when $|\overline{x}_2| < k < |\underline{x}_2|$, neither the Good-Only nor the Single-Bundle Posted Price mechanism can be optimal. 
The Ad-Tiered Posted Price mechanism, in contrast, places no restriction on $k$ and can therefore be optimal over a broader range of parameters. 
In \Cref{dis:uniform-bad}, we show that when $f$ is uniform in bad dimension, all three canonical mechanisms align neatly with the role of the third--party payment.

\noindent
Separately, the issue noted in \Cref{rem:convexity}—the inability to extend the necessary conditions to sufficiency—persists here. 
At the core of this limitation lies the difficulty of fully capturing two-dimensional convexity: 
for any value of~$k$, the structure of~$\mu$ ensures that in each canonical mechanism there remains at least one region in which dimensionality cannot be reduced.
Consequently, the sufficiency results continue to rely on the orthant conditions.
\end{remark}

We now present an example demonstrating optimality of the \emph{Ad-Tiered Posted Price} mechanism.

\begin{example}\label{ex:ad-tier}
\upshape
Let $X = [1,2] \times [-1,0]$, $\kappa(x) = 0.5$, and let $f$ be uniform on $X$. 
The transformed measure $\mu$ is
\begin{equation*}
\mu(A) 
= \mathbb{I}_{A}((1,-1)) 
+ \int_{\partial X} \mathbb{I}_{A}(x)\,\big((x_1, x_2 + 0.5)\!\cdot\!\hat{n}\big)\,dx 
- 3\int_{X} \mathbb{I}_{A}(x)\,dx.
\end{equation*}
Hence $\mu$ places a point mass $+1$ at $(1,-1)$, a total mass $-3$ uniformly over $X$, 
and has linear densities of $+0.5$ on $x_2=-1$, $+2$ on $x_1=2$, $+0.5$ on $x_2=0$, and $-1$ on $x_1=1$.

\noindent Let $p_{g}^*$ and $p_{sb}^*$ solve
\begin{align*}
1.5p_g^2 - 4p_g + 2p_{sb} + 1 &= 0,\\
1.5p_g^2 - 3p_gp_{sb} - 2.5p_g + 2p_{sb} + 2 &= 0,
\end{align*}
obtained by imposing $\mu(\mathcal{Y}) = \mu(\mathcal{Z}) = 0$. 
The unique feasible solution is $p_{g}^* \approx 1.12$ and $p_{sb}^* \approx 0.80$.
Condition \textup{(i)} is satisfied by construction.
We verify Conditions \textup{(ii), (iii)}, and \textup{(iv)} in \Cref{excal:ad-tier} and conclude that the \emph{Ad-Tiered Posted Price} mechanism is optimal in this example.
\end{example}

\subsection{Uniform Distribution: A Partial Characterization}\label{dis:uniform}

We next specialize our analysis to uniform distributions. 
When the density is uniform in the bad dimension, we show that optimality of the \emph{Ad-Tiered Posted Price} mechanism also imposes a restriction on the third--party payment $k$, complementing the necessary bounds already obtained for the Good-Only and Single-Bundle Posted Price mechanisms.
Under the stronger assumption that the density is uniform on the entire type space, we obtain sufficient conditions for optimality of each mechanism over explicit regions of $k$ and the support.
Together, these results provide a partial characterization of how the optimal mechanism changes with the third--party payment.

\subsubsection{Uniformity in the Bad Dimension: Necessary Bounds}\label{dis:uniform-bad}

We show that, under uniformity in the bad dimension, optimality of the \emph{Ad-Tiered Posted Price} mechanism also imposes a restriction on $k$.

\begin{prop}\label{prop:uniformbad}
If $f$ is uniform in the bad dimension $x_2$ and an \emph{Ad-Tiered Posted Price} mechanism is optimal, then
\[
|\overline{x}_2| < k \leq |\underline{x}_2|.
\]
\end{prop}

\begin{proof}
Suppose, for contradiction, that $k \le |\overline{x}_2|$.
Then only the right edge of $\mathcal{Y}$ carries strictly positive mass, while the top and left edges, as well as the interior, carry negative mass.
Consider a thin horizontal strip near the bottom boundary of $\mathcal{Y}$ (see \Cref{fig:ad-tier}),
\begin{equation*}
A := \{x \in \mathcal{Y} : x_2 \in [p_{g}^* - p_{sb}^*,\, p_{g}^* - p_{sb}^* + \varepsilon]\}, \qquad \text{for some } \varepsilon > 0.
\end{equation*}
Given the geometry of $\mathcal{Y}$, the uniformity of $f$ in the $x_2$ dimension, and the fact that $\mu(\mathcal{Y}) = 0$, it follows that $\mu(A) > 0$.

Now consider the convex, coordinatewise nonincreasing test function
\begin{equation*}
u_\varepsilon(x_1,x_2) := \exp\!\big(m[(p_{g}^* - p_{sb}^*+\varepsilon) - x_2]\big), \qquad m>0.
\end{equation*}
As in the proof of Proposition~\ref{prop:sbnec}, we can verify that 
$\int u_\varepsilon\, d\mu|_{\mathcal{Y}} > 0$, contradicting Lemma~\ref{lem:finitechar}.

When $k > |\underline{x}_2|$, a symmetric argument applied to region~$\mathcal{W}$ leads to the same contradiction.
Therefore, the stated range $|\overline{x}_2| < k \leq |\underline{x}_2|$ is necessary for the \emph{Ad-Tiered Posted Price} mechanism to be optimal.
\end{proof}

Taken together with the necessary conditions in \Cref{prop:good-only} and \Cref{prop:sbnec}, \Cref{prop:uniformbad} implies that, for $|\overline{x}_2|<k<|\underline{x}_2|$, the Ad-Tiered Posted Price mechanism is the only one of the three canonical mechanisms that is not ruled out by these necessary conditions.
For $k<|\overline{x}_2|$, the Good-Only Posted Price mechanism is the only one not ruled out, while for $k>|\underline{x}_2|$, the Single-Bundle Posted Price mechanism is the only one not ruled out.

\subsubsection{Uniformity on the Type Space: Sufficient Conditions}
\label{dis:uniform-both}

We now impose the stronger assumption that the density is uniform on the entire type space. 
Under this assumption, the orthant conditions in the preceding results can be verified explicitly, yielding sufficient conditions for optimality of each of the three mechanisms.

\begin{prop}\label{prop:uniform-suff}
Suppose that $f$ is uniform on the type space $X$. Then:

\begin{enumerate}
    \item A \emph{Good-Only Posted Price} mechanism is optimal if and only if
    $k\leq|\overline{x}_2|$.

    \item Suppose, in addition, that
    \[
    X=[a,a+c]\times[-c,0],
    \qquad a,c>0,\quad a\in[c/2,c].
    \]
    Then:
    \begin{enumerate}
        \item If $k\in[c/2,c]$, an \emph{Ad-Tiered Posted Price} mechanism is optimal.

        \item If $k\in[9c/8,3c]$, a \emph{Single-Bundle Posted Price} mechanism is optimal.
    \end{enumerate}
\end{enumerate}
\end{prop}

The three statements in Proposition~\ref{prop:uniform-suff} are established separately below. 
The Good-Only result provides an exact characterization, while the Ad-Tiered and Single-Bundle results give sufficient conditions on the support and the third-party payment.

\paragraph{Good-Only Posted Price.}
The Good-Only case admits an exact characterization.

\begin{lemma}\label{lem:uniform-good-only}
If $f$ is uniform on $X$, then a Good-Only Posted Price mechanism is optimal if and only if $k\leq|\overline{x}_2|$.
\end{lemma}

The proof verifies the orthant condition in \Cref{theo:good-only} explicitly; see Appendix~\ref{app:unigood}.

\paragraph{Single-Bundle Posted Price.}
We next consider the Single-Bundle Posted Price mechanism. 
Under a uniform density, we obtain a sufficient condition for optimality over a class of rectangular supports and sufficiently high third--party payments.

\begin{lemma}\label{lem:uniform-single-bundle}
Suppose that $f$ is uniform on $X:=[a,a+c]\times[-c,0]$, where $a,c>0$ and $a\in[c/2,c]$. 
If the third-party payment $k\in[\frac{9}{8}c,3c]$, then a Single-Bundle Posted Price mechanism is optimal.
\end{lemma}

The proof verifies the upper-right orthant condition in \Cref{theo:sbsuff}; see Appendix~\ref{app:unisingle}.

\paragraph{Ad-Tiered Posted Price.}
Finally, we consider the Ad-Tiered Posted Price mechanism. 
For the same class of rectangular supports, the orthant conditions can be verified over an intermediate range of third--party payments.

\begin{lemma}\label{lem:uniform-ad-tier}
Suppose $f$ is uniform on $X := [a,a+c] \times [-c,0]$, where $a,c > 0$ and $a \in [c/2,c]$.
If the third-party payment $k \in [c/2,c]$, then an Ad-Tiered Posted Price
mechanism is optimal.
\end{lemma}

The proof verifies the lower-right and upper-right orthant conditions in \Cref{theo:ad-tier}; see Appendix~\ref{app:uniadtier}.
Next, we present an example where none of the three mechanisms discussed above are optimal.

\begin{example}\label{ex:none}
\upshape
Suppose $f$ is uniform on $X = [1/4,5/4] \times [-1,0]$ and $k = 1/5$.
Borrowing analysis from proof of \Cref{lem:uniform-ad-tier}, if an Ad-Tiered Posted Price mechanism were optimal, there must exist $l,m \geq 0$ with $l+m \leq 1$ such that \Cref{eq:yzero,eq:wzero} are satisfied simultaneously.
Substituting $a=1/4$, $c=1$ and $k=1/5$ into \Cref{eq:yzero}, we get $30l^2-35l-40m+4=0$, so that $m=\frac{30l^2-35l+4}{40}$.
Since $m\geq0$, we must have $30l^2-35l+4\geq0$, which, together with $0 \leq l \leq 1$, implies $l \in [0,\frac{35-\sqrt{745}}{60}]$.

Substituting for $m$ into \Cref{eq:wzero} yields $p(l) := 360l^3-150l^2-339l+124=0$.
Observe that $p'(l) < 0$ in the relevant interval so $p$ is strictly decreasing.
At the end point, $p(\frac{35-\sqrt{745}}{60}) > 0$, so we have a contradiction.
Hence there are no $l,m\geq0$ satisfying \Cref{eq:yzero,eq:wzero} simultaneously, and therefore the Ad-Tiered Posted Price cannot be optimal in this example (\Cref{prop:ad-tier}).
The other two mechanisms are ruled out by necessary conditions on $k$ in \Cref{prop:good-only,prop:sbnec}.
\end{example}


\subsection{Economic Interpretation of the Three Regimes}

The economic role of the third–party payment $k$ is to change the seller’s optimal allocations in a two-dimensional type space. 
The two dimensions capture the buyer’s valuation for the good and the buyer’s preference for avoiding the bad. 
As $k$ increases, the seller’s tradeoff between allocating the bad and extracting revenue from the buyer changes, leading to different optimal allocations across the type space. 
The transformed measure $\mu$ captures these forces jointly: it incorporates the two-dimensional distribution of types, the effect of the third–party payment on utility, and the resulting effect of the allocation on the seller’s overall revenue.

The orthant conditions provide a geometric way to interpret these changing allocation forces. 
An orthant represents a group of types that are jointly high or low along the two dimensions—for example, types that are high in one dimension and low in the other—and its $\mu$-mass measures the net effect of changing the allocation for that group on the transformed objective.
Thus, the orthant conditions determine whether the seller benefits from expanding or reducing the allocation at different thresholds in the type space. 
As $k$ increases, the resulting changes in these marginal effects give rise to the progression across the three regimes below.

\begin{enumerate}
\item {\bf Low $k$: Good-Only Regime}
\subitem \underline{Geometry}: At low $k$, $\mu$ assigns positive mass to the bottom edge (ad-averse types) and negative mass to the top edge (ad-tolerant types).

\subitem \underline{Interpretation}:~
When the third–party payment is small, the optimal allocation favors withholding the bad.
The right-bottom orthant condition in \Cref{theo:good-only}, requiring nonnegative $\mu$-mass for types with sufficiently high valuation for the good and willingness to pay to avoid the bad, provides the relevant sufficient condition.
In particular, positive $\mu$-mass on every such orthant means that increasing the indirect utility in the right-bottom direction weakly raises the transformed objective.
Since the derivatives of $u$ correspond to allocations, maximizing in the good direction gives allocation $1$, while maximizing in the downward bad direction gives the minimum feasible allocation, $0$.
Thus, the orthant condition favors $(1,0)$.

The negative mass at the top edge has a complementary role in the necessary condition established in \Cref{prop:good-only}. 
Since $\mu(\mathcal Y)=0$, positive mass at the top edge would necessarily be offset by negative mass elsewhere in $\mathcal Y$, and would be incompatible with the nonnegativity of all relevant right-bottom orthants. 

\subitem \underline{Outcome}:~
At sufficiently low $k$, the seller optimally withholds the bad and sells only the good.

\item {\bf Intermediate $k$: Ad-Tiered Regime}

\subitem \underline{Geometry}:~
Both the top and bottom edges of $X$ carry positive $\mu$-mass.

\subitem \underline{Interpretation}:~
As $k$ increases, the third–party revenue from allocating the bad becomes more important. 
The positive mass at the top indicates that the transformed objective can now support allocating the bad to some types, while the positive mass at the bottom continues to reflect the value of preserving utility for ad-averse types. 
The orthant conditions in \Cref{theo:ad-tier} formalize this tension and favor separating types according to their preference for avoiding the bad.

\subitem \underline{Outcome}:~
Among types with sufficiently high valuation for the good, ad-averse types receive the allocation $(1,0)$, providing the good without the bad, while ad-tolerant types receive $(1,1)$, providing both the good and the bad. 
Among the canonical mechanisms we consider, only the Ad-Tiered Posted Price mechanism can be optimal in this regime.

\item {\bf High $k$: Single-Bundle Regime}

\subitem \underline{Geometry}:~
The top edge of $X$ carries positive $\mu$-mass, while the bottom edge carries negative $\mu$-mass.

\subitem \underline{Interpretation}:~
As $k$ increases further, the third–party revenue from allocating the bad becomes sufficiently important to support allocating the bad throughout the relevant region of the type space. 
The sign pattern of $\mu$ is now reversed relative to the Good-Only regime, while the upper-right orthant condition in \Cref{theo:sbsuff} formalizes the resulting allocation forces. 
In particular, nonnegative $\mu$-mass on every relevant upper-right orthant means that increasing either allocation coordinate weakly raises the transformed objective. 
Since both coordinates are bounded above by one, the optimal allocation reaches both upper bounds.

\subitem \underline{Outcome}:~
The seller therefore offers a single ad-supported bundle, allocating $(1,1)$ to types with sufficiently high valuation for the good and low ad-aversion. 
Optimal mechanism does not separate ad-averse and ad-tolerant buyers through different levels of ad exposure, the seller allocates the bad to all types who receive the good. 
Among the canonical mechanisms we consider, only the Single-Bundle Posted Price mechanism can be optimal in this regime.
\end{enumerate}

\paragraph{Uniform example.}
The uniform example illustrates the economic progression described above in a particularly transparent way. 
Holding the type distribution fixed and varying only the third–party payment $k$, the seller’s optimal allocation can move from withholding the bad, to using ad-tiered menu, and eventually to allocating the bad together with the good. 
The uniform distribution results establish these regimes over explicit ranges of $k$, while the remaining gaps reflect the fact that our sufficient conditions need not be tight.


\section{Discussion}\label{sec:discussion}

\subsection{MM Assumption: $k$-Shift}\label{dis:regularity}

We extend the standard \cite{mcafee_mcmillan_1988} condition to
\begin{equation*}
3f(x) + x \cdot \nabla f(x) + k\,\partial_{2}f(x) \ge 0, \qquad k \ge 0,
\end{equation*}
which reduces to the familiar form $3f + x \cdot \nabla f \ge 0$ when $k=0$.
The additional $k$-term modifies the requirement according to the sign of the slope in the second coordinate:
it \emph{relaxes} the inequality where $\partial_{2}f>0$ and \emph{tightens} it where $\partial_{2}f<0$.
The following example illustrates this dependence.

\begin{example}[Truncated log--linear family]
\upshape
Let 
\begin{equation*}
f(x_1,x_2)=C\,e^{-a x_1+b x_2},\qquad (x_1,x_2)\in X=[0,A]\times[-B,0],
\end{equation*}
for constants $A,B,C>0$.
The gradient is 
\begin{equation*}
\nabla f(x_1,x_2)=(-a f,\, b f),
\quad\text{and hence}\quad 
\partial_{2}f = b f.
\end{equation*}
Substituting into the MM condition yields
\begin{equation*}
3f + x \!\cdot\! \nabla f + k\,\partial_{2}f
  = (3 - a x_1 + b x_2 + k b)\,f(x_1,x_2).
\end{equation*}
Because $f>0$, the inequality holds throughout $X$ whenever 
\begin{equation*}
3 - aA - bB + k b \ge 0,
\end{equation*}
which is strictly easier to satisfy than the baseline case $(k=0)$ whenever $b>0$.
Conversely, when $b<0$, the requirement becomes more stringent.
\end{example}

\subsection{Weakening Constant Third--Party Payment}\label{dis:k}
While we argued that the duality and finite--menu characterization of DDT extend to our setting for a general third--party payment function $\kappa(x)$, our analysis of the canonical mechanisms assumed a constant payment $\kappa(x)=k$.
This assumption can be substantially relaxed.

In particular, the proofs of \Cref{theo:good-only} and subsequent results rely only on the signs of $\mu$ along the horizontal boundaries of $X$.
For instance, in \Cref{theo:good-only}, we require that the top edge carry negative and the bottom edge positive mass.
For a general function $\kappa(x)$, it suffices that
\begin{equation*}
(\overline{x}_2+\kappa(x_1,\overline{x}_2))\,f(x_1,\overline{x}_2) \le 0 
\quad\text{and}\quad
(\underline{x}_2+\kappa(x_1,\underline{x}_2))\,f(x_1,\underline{x}_2) \ge 0
\quad\text{for all }x_1,
\end{equation*}
ensuring, respectively, top--edge negativity and bottom--edge positivity.
Analogous sign conditions appear in the other canonical cases as well.

Thus the constant--$k$ specification mainly serves expositional clarity: it keeps the sign of $\mu$ same throughout the top (or bottom) edge, corresponding to the bad dimension $x_2$, while capturing the essential economic interpretation of ad--monetization. 
It also reflects the prevailing structure of subscription markets, where advertiser payments depend on impressions rather than on users’ ad aversion or willingness to pay for the service.

\subsection{Comparison with \citet{Ichihashi_Jeon_Kim_2026}}\label{sec:oneD}
To make the comparison precise, we first present a simpler version of their {\it one-dimensional} model in our notation.
Let $x \in X := [\underline{x},\overline{x}] \subset \mathbb{R}_+$ denote the user’s private information (type), drawn from distribution $F$ with density $f$.
Let $d\in(0,1)$ denote the relative nuisance from advertising, which is {\it public} information.
The platform’s mechanism is described by $q_1,q_2:X\to[0,1]$, denoting the allocations of the service and advertising, respectively, and a payment function $t:X\to\mathbb{R}$.
A user of type $x$ who reports truthfully obtains utility
\[
x (q_1(x) - dq_2(x)) - t(x).
\]

In addition to payments from users, the platform receives third-party revenue $rq_2(x)$, where $r>0$. 
Writing $q=q_1-dq_2$ and $\phi(x)=x-\frac{1-F(x)}{f(x)}$,
the platform’s problem, using standard one-dimensional envelope argument, reduces to
$\max_{q_1,q_2}\int_X [\phi(x)q(x)+rq_2(x)]f(x)\,dx$
subject to $q(x)$ being nondecreasing and $q_1(x) - d q_2(x) \geq 0$ (an additional constraint they impose, we discuss this below). 
When $\phi$ is strictly increasing, it is straightforward to see that the optimal allocation has three regions:
\[
(q_1(x),q_2(x))=
\begin{cases}
(1,0), & x\ge \phi^{-1}(r/d),\\
(1,1), & x\in(\phi^{-1}(0),\phi^{-1}(r/d)),\\
(d,1), & x\le \phi^{-1}(0).
\end{cases}
\]

The three-outcome menu above reflects the platform’s two-dimensional allocation problem rather than multidimensional preference heterogeneity. 
In their model, service valuation and advertising nuisance are linked through a single type, so users remain ordered along one dimension even though the platform can choose both service and advertising. 
In our model, by contrast, service valuation and ad aversion constitute distinct dimensions of private information, allowing users with sufficiently high valuation for the good to be separated by their tolerance for advertising: more ad-averse users can receive $(1,0)$ while more ad-tolerant users receive $(1,1)$. 
Thus, advertising can serve as a genuine screening margin between users who are otherwise similar in their valuation for the good—a form of separation that is absent from the one-dimensional formulation.

There is also an important difference in the feasibility restrictions imposed on allocations.
In their full model, the service and advertising levels are discretized and ordered as different ``qualities,’’ and feasible service–advertising pairs are restricted so that every such pair provides nonnegative utility to every user type, abstracting from transfers. In our continuous formulation above, this restriction appears as $q_1(x)-dq_2(x)\geq 0$.
This restriction is consequential: it prevents an allocation involving advertising from generating negative gross utility for {\it any} type. Our model imposes no analogous restriction. An allocation chosen by some types may generate negative utility for other types who do not choose it—for example, an ad-averse type may dislike an ad-supported allocation that is nevertheless optimal for a more ad-tolerant type. 
Thus, the one-dimensional formulation’s restriction limits which service–advertising combinations can coexist in the menu. 
The example above illustrates that this distinction is substantive: the constraint binds at one of the allocations.

Finally, the one-dimensional structure in \citet{Ichihashi_Jeon_Kim_2026} is well suited to their analysis of content quality and the effects of advertising profitability and related policies. 
Our focus is different: we use the multidimensional type space to study how changes in the third–party payment reshape the platform’s optimal allocation across users with different valuations for the good and aversion to the bad. 
In particular, our results illustrate how the optimal mechanism can move from withholding the bad, to using it as a screening dimension, and eventually to allocating it together with the good as third–party revenue increases.

\subsection{Scope, Limitations, and Computational Simplicity}\label{dis:scope}
DDT characterize optimal finite-menu mechanisms through integral inequalities that must hold for all convex, monotone test functions---a condition that, while exact, is not computationally tractable in general. 
Our results build directly on this foundation but illuminate where the main source of intractability lies even in simple environments. 
As highlighted in \Cref{rem:convexity}, the difficulty arises from the need to verify convex dominance in fully two-dimensional regions, where the cone of convex functions lacks a finite or countable generating family.

Nonetheless, our analysis provides meaningful progress in this direction. 
In the Good-Only Posted Price mechanism, we explicitly identify a region ($\mathcal{Z}$) where the dominance conditions reduce to a one-dimensional family of hinge functions, making verification both interpretable and exact. 
For other regions, such as $\mathcal{W}$ and $\mathcal{Y}$, we replace exhaustive convex testing with simple orthant conditions on $\mu$. 
Although these conditions are stronger than necessary, they are readily computable for densities that admit closed-form expressions and still capture the essential geometry of optimality. 
In this sense, our results retain the spirit of DDT: 
for distributions with simple analytical structure, one can recover clear and interpretable optimal mechanisms without relying on numerical optimization or the full convex program.

\subsection{Extensions Beyond Two Dimensions}\label{dis:higherdim}
The results of DDT extend to environments with any number of dimensions, allowing multiple goods and bads. 
Our framework therefore has immediate analogues in richer market settings. 
For instance, a streaming platform such as Disney+ may bundle several ``goods'' (e.g., Hulu, ESPN+, HBO~Max) under a unified subscription, while the ``bad'' remains a single nuisance such as advertising exposure. 
Conversely, digital ecosystems often feature multiple ``bads''---for example, both advertisements and privacy loss---each generating third--party revenue for the seller.

Extending our results to such multi-dimensional environments is, however, not straightforward. 
While the duality and finite-menu characterization of DDT remain valid for any dimension, our proofs of sufficiency rely crucially on geometric arguments specific to the two-dimensional structure of~$\mu$.
Establishing analogous geometric sufficiency conditions in higher dimensions is left for future work.


\section{Conclusion}\label{sec:conclusion}

We show that introducing a bad and third–party payments into an otherwise standard multi–good framework leads to qualitatively new outcomes.
Even this minimal extension reshapes the transformed measure, alters which incentive constraints bind, and yields a clear variation of optimal mechanisms—Good–Only, Ad–Tiered, and Single–Bundle—as ad revenue varies.
Adapting the duality framework of DDT, we derive tractable geometric conditions—orthant tests—that render optimality both transparent and easy to verify.
The results illustrate how small structural changes in screening models can generate rich comparative statics and economically intuitive mechanisms.

\bibliographystyle{ecta}
\bibliography{order}

\appendix

\section{Appendix : Omitted Proofs}

\subsection{Derivation of Expression \eqref{objDiv}}\label{sec:intbyparts}

\begin{align}
\int_X (x_i + {v_{\kappa(x)}}_i)\,\frac{\partial u}{\partial x_i}\, f(x)\,dx
&= \int_{X_{-i}} \int_{X_i} \frac{\partial u}{\partial x_i}\,(x_i + {v_{\kappa(x)}}_i)\, f(x)\,dx_i\,dx_{-i} \notag\\
&= \int_{X_{-i}} \Big\{u(x)\,(x_i + {v_{\kappa(x)}}_i)\,f(x)\Big|_{\underline{x_i}}^{\overline{x_i}}
   - \int_{X_i} u(x)\,\frac{\partial}{\partial x_i}\!\big[(x_i + {v_{\kappa(x)}}_i)\, f(x)\big]dx_i \Big\} dx_{-i} \notag\\
&= \int_{X_{-i}} 
   \Big\{u(x)\,(x_i + {v_{\kappa(x)}}_i)\,f(x)\Big|_{\underline{x_i}}^{\overline{x_i}}\Big\} dx_{-i} \notag\\
&\quad - \int_X u(x)\,\Big[(1+\tfrac{\partial {v_{\kappa(x)}}_i}{\partial x_i})\,f(x)
   + (x_i + {v_{\kappa(x)}}_i)\,\tfrac{\partial f}{\partial x_i}\Big]dx.
\label{eq:ibp_coord}
\end{align}

Summing \eqref{eq:ibp_coord} over $i=1,2$, and noting that 
$\tfrac{\partial {v_{\kappa(x)}}_1}{\partial x_1} = 0$ and $\tfrac{\partial {v_{\kappa(x)}}_2}{\partial x_2} = \partial_2 \kappa(x)$, we obtain
\begin{align*}
\int_X \big[(x + v_{\kappa(x)})\!\cdot\!\nabla u(x)\big] f(x)\,dx 
&= 
\int_{\partial X} u(x)\,f(x)\,\big((x + v_{\kappa(x)})\!\cdot\!\hat n\big)\,dx \notag\\
&\quad - \int_X u(x)\,\Big[\nabla f(x)\!\cdot\!(x + v_{\kappa(x)}) 
   + \big(2 + \partial_2 \kappa(x)\big)f(x)\Big]dx.
\end{align*}

Adding and subtracting 
$\int_X (u(x) - u(\underline{x}))\,f(x)\,dx$ on both sides, we have
\begin{align*}
\int_X \big[(x + v_{\kappa(x)})\!\cdot\!\nabla u(x) 
   - (u(x)-u(\underline{x}))\big] f(x)\,dx
&= u(\underline{x})
+ \int_{\partial X} u(x)\,f(x)\,\big((x + v_{\kappa(x)})\!\cdot\!\hat n\big)\,dx \notag\\
&\quad - \int_X u(x)\,\Big[\nabla f(x)\!\cdot\!(x + v_{\kappa(x)})
   + \big(3 + \partial_2 \kappa(x)\big)f(x)\Big]dx.
\end{align*}


\subsection{Proofs for Good-Only Posted Price}

\subsubsection{Proof of Proposition~\ref{prop:good-only}}\label{appendix:good-only}

\noindent(i) Let $R \subseteq X$ denote the region of types selecting any given menu option. 
By \Cref{lem:finitechar}, $\mu|_R \preceq_{cvx(\vec v)} 0$, which implies $\int_R u \, d\mu \leq 0$ for all convex, $\vec v$-monotone functions $u$. 
Since constant functions are convex and $\vec v$-monotone for any $\vec v$, taking $u \equiv 1$ and $u \equiv -1$ yields $\mu(R) = 0$. 
In particular, $\mu(\mathcal{Z}) = \mu(\mathcal{W}) = 0$.

\medskip

\noindent(ii) Suppose, for contradiction, that $k > |\overline{x}_2|$. 
From the expression \eqref{eq:transk}, the top edge 
\begin{equation*}
T := \{x \in X : x_2 = \overline{x}_2\}
\end{equation*}
then has strictly positive linear density $(\overline{x}_2 + k)f(x)$. 
Hence $\mu(T \cap \mathcal{Z}) > 0$.
Note that since the allocation in $\mathcal{Z}$ is $(0,0)$, the test functions required by \Cref{lem:finitechar} are convex and coordinate-wise nondecreasing.
Consider the family of test functions
\begin{equation*}
u_\delta(x) := \exp\!\left(\tfrac{x_2 - \overline{x}_2}{\delta}\right), \qquad \delta > 0.
\end{equation*}
Each $u_\delta$ is convex and nondecreasing. 
As $\delta \to 0$, we have $u_\delta(x) \to 1$ for $x \in T \cap \mathcal{Z}$ and $u_\delta(x) \to 0$ for $x \in \mathcal{Z} \setminus T$. 
Thus,
\begin{equation*}
\int u_\delta \, d\mu|_{\mathcal{Z}} \;\longrightarrow\; \mu(T \cap \mathcal{Z}) > 0,
\end{equation*}
contradicting \Cref{lem:finitechar}, which requires $\int u \, d\mu|_{\mathcal{Z}} \leq 0$ for all convex, nondecreasing $u$. 
Therefore $k \leq |\overline{x}_2|$.

\medskip

\noindent(iii) Fix $t \in [\underline{x}_1,p_g^*]$ and consider the hinge test function
\begin{equation*}
u_t(x) := (x_1 - t)_+,
\end{equation*}
which is convex and coordinate-wise nondecreasing. By \Cref{lem:finitechar},
\begin{equation}\label{eq:goodnec}
\int u_t \, d\mu|_{\mathcal{Z}} \;\leq\; 0.
\end{equation}
Define the one-dimensional marginal
\begin{equation*}
\mu_1(A) := \mu\big(A \times [\underline{x}_2,\overline{x}_2]\big),
\quad A \subseteq [\underline{x}_1,\overline{x}_1]\ \text{measurable}.
\end{equation*}
Since $u_t$ depends only on $x_1$,
\begin{equation*}
\int u_t(x) \, d\mu|_{\mathcal{Z}} \;=\; \int_{[\underline{x}_1,p_{g}^*]} (x_1 - t)_+ \, d\mu_1(x_1).
\end{equation*}
Moreover, for each $x_1 \in [\underline{x}_1,p_g^*]$,
\begin{equation*}
(x_1 - t)_+ \;=\; \int_t^{p_g^*} \mathbb{I}_{\{x_1 > z\}} \, dz.
\end{equation*}
Substituting yields
\begin{equation}\label{eq:good-only}
\int u_t \, d\mu|_{\mathcal{Z}}
\;=\; \int_{[\underline{x}_1,p_{g}^*]} \!\int_t^{p_g^*} \mathbb{I}_{\{x_1 > z\}} \, dz \, d\mu_1(x_1)
\;=\; \int_t^{p_g^*} \mu_1\big((z,p_g^*]\big)\,dz
\;=\; -\int_t^{p_g^*} M(z)\,dz,
\end{equation}
where in the second equality we use Fubini’s theorem, and in the third we use $M(x_1):=\mu\big([\underline{x}_1,x_1]\times[\underline{x}_2,\overline{x}_2]\big)$ together with $M(p_g^*)=0$, which implies $\mu_1((z,p_g^*])=M(p_g^*)-M(z)=-M(z)$. 
Combining \eqref{eq:goodnec} and \eqref{eq:good-only} yields Condition \textup{(iii)}.

\subsubsection{Proof of \Cref{theo:good-only}}\label{appendix:good-only2}
\noindent {\bf Sufficiency for $\mathcal{Z}$.}
Let $B := [\underline{x}_1,p_{g}^*] \times \{\underline{x}_2\}$ denote the bottom edge of $\mathcal{Z}$.
By MM, Condition (i) of Proposition \ref{prop:good-only} and the fact that the left edge of $\mathcal{Z}$ carries negative $\mu$-mass, the measure $\mu$ assigns nonpositive mass to all points in $\mathcal{Z} \setminus B$.

Fix any $u$ that is convex and nondecreasing on $\mathcal{Z}$, and define
\begin{equation*}
u'(x_1) := u(x_1,\underline{x}_2), \qquad x_1 \in [\underline{x}_1,p_{g}^*].
\end{equation*}
Note that $u'$ is convex and nondecreasing.

Since u is nondecreasing in both coordinates, we have $u(x_1,x_2)\ge u’(x_1)$ for all $(x_1,x_2) \in \mathcal{Z}.$
Consequently, integrating against $\mu$ and using its sign structure gives
\begin{equation}\label{eq:good-only-zsuff}
\int u \, d\mu|_{\mathcal Z}
\;\leq\;
\int_{[\underline{x}_1,p_{g}^*]} u'(x_1) \, d\mu_1(x_1),
\end{equation}
where $\mu_1$ denotes the marginal of $\mu$ on the $x_1$-coordinate.

Condition (iii) of Proposition~\ref{prop:good-only} and Equation~\eqref{eq:good-only} imply that
\begin{equation*}
\int_{[\underline{x}_1,p_{g}^*]} (x_1-t)_+\,d\mu_1(x_1) \le 0 \qquad \text{for all } t \in [\underline{x}_1,p_{g}^*].
\end{equation*}
By Theorem 1.5.7 of \citet{muller2002comparison}, this implies
\begin{equation*}
\int_{[\underline{x}_1,p_{g}^*]} v(x_1)\,d\mu_1(x_1)\le 0
\end{equation*}
for any convex, nondecreasing $v:B \to \mathbb{R}$.  
Therefore, Inequality \eqref{eq:good-only-zsuff} implies $\int u\,d\mu|_{\mathcal{Z}}\le 0$, as required by Lemma~\ref{lem:finitechar} for sufficiency on $\mathcal{Z}$

\medskip

\noindent \noindent {\bf Sufficiency for $\mathcal{W}$.}
The relevant test functions on $\mathcal W$ are convex, nonincreasing in $x_1$ and 
nondecreasing in $x_2$, so their sublevel sets are precisely the \emph{lower sets} defined below.

For a measurable set $\mathcal L\subseteq\mathcal W$, call $\mathcal L$ \emph{lower set} if 
$(x_1,x_2)\in\mathcal L$ and $(x_1',x_2')\in\mathcal W$ with $x_1' \geq x_1$ and $x_2' \leq x_2$
imply $(x_1',x_2')\in\mathcal L$.

Given a lower set $\mathcal L$, define its associated lower–right orthant
\begin{align*}
O_{\mathcal L} &:= [x_{1\mathcal L},\overline x_1]\times[\underline x_2,x_{2\mathcal L}],\\
x_{1\mathcal L} &:= \inf\{x_1:(x_1,\underline x_2)\in\mathcal L\}, \\
x_{2\mathcal L} &:= \sup\{x_2:(\overline x_1,x_2)\in\mathcal L\}.
\end{align*}
Then $\mathcal L\subseteq O_{\mathcal L}\subseteq\mathcal W$.

By MM, Condition (ii) of Proposition \ref{prop:good-only} , the added region $O_{\mathcal L}\setminus\mathcal L$ has nonpositive $\mu$–mass,
hence
\begin{equation*}
 \mu|_{\mathcal W}(O_{\mathcal L}) \le \mu|_{\mathcal W}(\mathcal L).
\end{equation*}
By Condition (iv), every lower–right orthant in $\mathcal W$ has nonnegative mass,
hence $\mu|_{\mathcal W}(O_{\mathcal L})\ge 0$. Therefore
\begin{equation}\label{eq:SEorthant}
\mu|_{\mathcal W}(\mathcal L)\ \ge\ 0.
\end{equation}

Now we adapt the layer–cake step of \cite{muller2002comparison}, Theorem 3.3.4, to our lower sets.
Let $u:\mathcal W\to\mathbb R$ be bounded, nonincreasing in $x_1$ and nondecreasing in $x_2$, and set $\gamma := \sup_{x \in \mathcal{W}} u(x)$.
For any real $\alpha$, define the lower level set $L_{\alpha} := \{x \in \mathcal{W}: u(x) \leq \alpha \}$.
Each $L_{\alpha}$ is a lower set, so by \eqref{eq:SEorthant} we have $\mu|_{\mathcal{W}}(L_{\alpha}) \geq 0$.

For $n \in \mathbb{N}$, define
\begin{equation*}
u_n(x) := \gamma - \frac{1}{n}\sum_{k=0}^{\infty} \mathbb{I}_{L_{\gamma - k/n}}(x)
\end{equation*}

For fixed $x$, only those $k$ with $k/n \le \gamma-u(x)$ contribute to the sum, so $u_n(x)$ is well defined and finite.
Moreover, by construction,
$0 \le u(x)-u_n(x) < \tfrac{1}{n}$, for all $x\in\mathcal W$.
So $u_n\to u$ uniformly on $\mathcal{W}$.

Now integrate against $\mu|_{\mathcal{W}}$. 
\begin{equation*}
\int u_n \, d\mu|_{\mathcal{W}} = \gamma \mu|_{\mathcal{W}}(\mathcal{W}) - \frac{1}{n}\sum_{k=0}^{\infty} \mu|_{\mathcal{W}}(L_{\gamma-k/n})
= -\frac{1}{n}\sum_{k=0}^{\infty} \mu|_{\mathcal{W}}(L_{\gamma-k/n}) \leq 0.
\end{equation*}

The second equality follows from $\mu|_{\mathcal{W}}(\mathcal{W}) = 0$ (by Condition (i) of Proposition \ref{prop:good-only}), while the inequality uses that $\mu|_{\mathcal{W}}(L_{\alpha}) \ge 0$ for all lower sets $L_{\alpha}$.

Finally, since $u_n \to u$ uniformly and $\mu|_{\mathcal {W}}$ is a finite signed measure, we can pass to the limit under the integral:
\begin{equation*}
\int u \, d\mu|_{\mathcal{W}} = \lim_{n\to\infty} \int u_n \, d\mu|_{\mathcal {W}} \;\le 0.
\end{equation*}

By \Cref{lem:finitechar}, this establishes the required sufficiency condition for $\mathcal W$.


\subsection{Proofs for Single-Bundle Posted Price}

\subsubsection{Proof of Proposition \ref{prop:sbnec}}\label{appendix:sbnec}

On the right edge
\begin{equation*}
R := \{\overline{x}_1\}\times[\underline{x}_2,\overline{x}_2],
\end{equation*}
the measure $\mu$ has line density $\overline{x}_1 f(x)$, which is strictly positive. 
Therefore, note that $\mu(A)>0$ for all non-empty measurable $A\subseteq R$.

We proceed by cases according to the parameter $k$.

\noindent\textbf{Case 1: $k < |\underline{x}_2|$.}

\underline{Case 1a: $p_{sb}^* > \overline{x}_1 + \underline{x}_2$.}  
Then the vertical segment
\begin{equation*}
A := \{\overline{x}_1\}\times[\underline{x}_2,\,p_{sb}^*-\overline{x}_1] \subseteq R
\end{equation*}
lies in $\mathcal{Z}$ and satisfies $\mu(A)>0$. 
But by Lemma~\ref{lem:finitechar}, $\int u\,d\mu|_{\mathcal{Z}} \leq 0$ for all convex, nondecreasing $u$. 
Take
\begin{equation*}
u_\delta(x) := \exp\!\left(\frac{x_1-\overline{x}_1}{\delta}\right),\quad \delta>0.
\end{equation*}
This function is convex and nondecreasing. 
Moreover, $\int u_\delta\,d\mu|_{\mathcal{Z}} \to \mu(A)>0$ as $\delta\to 0$, a contradiction. 
Hence this sub-case is impossible.

\underline{Case 1b: $p_{sb}^* < \overline{x}_1 + \underline{x}_2$.}

On $B$
\begin{equation*}
B := [\underline{x}_1,\overline{x}_1]\times\{\underline{x}_2\},
\end{equation*}
the measure $\mu$ has line density $-(\underline{x}_2+k)f(x)$. 
Since $\underline{x}_2+k<0$ in this case, $\mu(A)>0$ for all measurable $A\subseteq B$.
  
The horizontal segment 
\begin{equation*}
A := [\,p_{sb}^*-\underline{x}_2,\;\overline{x}_1]\times\{\underline{x}_2\} \subseteq B
\end{equation*}
lies in $\mathcal{Y}$ in this sub-case and satisfies $\mu(A)>0$. 
By Lemma~\ref{lem:finitechar}, $\int u\,d\mu|_{\mathcal{Y}} \leq 0$ for all convex, nonincreasing $u$. 
Consider
\begin{equation*}
u_\delta(x) := \exp\!\left(\frac{\underline{x}_2-x_2}{\delta}\right).
\end{equation*}
This function is convex and nonincreasing, and $\int u_\delta\,d\mu|_\mathcal{Y} \to \mu(A)>0$ as $\delta\to 0$, again a contradiction. 
Thus this sub-case is also impossible.

\underline{Case 1c: $p_{sb}^* = \overline{x}_1 + \underline{x}_2$.}

On $B$, the line density is $-(\underline{x}_2 + k)f(x_1,\underline{x}_2)>0$. 
By continuity of $f$ there exist $m_0>0$ and $\varepsilon_0>0$ such that 
$-(\underline{x}_2 + k)f(x_1,\underline{x}_2) \geq m_0$ for all 
$x_1\in[\overline{x}_1-\varepsilon_0,\overline{x}_1]$.
For $0<\varepsilon<\varepsilon_0$, define the corner region
\begin{equation*}
\mathcal{C}_{\varepsilon} 
:= \big([\overline{x}_1 - \varepsilon,\,\overline{x}_1] \times [\underline{x}_2,\,\underline{x}_2 + \varepsilon]\big)\cap \mathcal{Z}.
\end{equation*}
Since the interior density of $\mu$ is bounded by some constant $C>0$, its contribution over $\mathcal{C}_\varepsilon$ is at most $C \varepsilon^{2}$ in magnitude, while the bottom-edge contribution is at least $m_0\varepsilon$.
Hence
\begin{equation*}
\mu(\mathcal{C}_\varepsilon) \ge m_0\,\varepsilon - C\,\varepsilon^2 > 0
\quad \text{for sufficiently small }\varepsilon>0.
\end{equation*}

Consider the admissible test function
\begin{equation*}
u_\varepsilon(x_1,x_2) := \exp\!\big(m[x_1-(\overline{x}_1-\varepsilon)]\big), \qquad m>0.
\end{equation*}
Then $u_\varepsilon$ is convex and coordinate-wise nondecreasing, satisfies $u_\varepsilon \ge 1$ on $\mathcal C_\varepsilon$, and
\begin{equation*}
u_\varepsilon(x_1,x_2) \le e^{-m\varepsilon}
\quad \text{for all } (x_1,x_2)\in \mathcal Z\setminus \mathcal C_\varepsilon.
\end{equation*}
Decomposing $\mathcal Z=\mathcal C_\varepsilon\cup(\mathcal Z\setminus\mathcal C_\varepsilon)$ gives
\begin{equation*}
\int u_\varepsilon\,d\mu|_{\mathcal Z} \geq
\int u_\varepsilon\,d\mu|_{\mathcal C_\varepsilon}
+
\int u_\varepsilon\,d\mu|_{\mathcal Z\setminus\mathcal C_\varepsilon}
\ge
\mu(\mathcal C_\varepsilon)
-
e^{-m\varepsilon}|\mu(\mathcal Z\setminus\mathcal C_\varepsilon)|.
\end{equation*}
As shown above, $\mu(\mathcal C_\varepsilon) > 0$ for sufficiently small $\varepsilon>0$, while $|\mu(\mathcal Z\setminus\mathcal C_\varepsilon)|<\infty$. 
Choosing $m$ large enough so that $e^{-m\varepsilon}|\mu(\mathcal Z\setminus\mathcal C_\varepsilon)| < \mu(\mathcal C_\varepsilon)$ yields
\begin{equation*}
\int u_\varepsilon\,d\mu|_{\mathcal Z} > 0,
\end{equation*}
contradicting Lemma~\ref{lem:finitechar}.
Hence this sub-case cannot occur.
Since sub-cases 1a–1c exhaust all possibilities when $k < |\underline{x}_2|$, this entire case is ruled out.

\noindent\textbf{Case 2: $k \geq |\underline{x}_2|$.} 

\underline{Case 2a: $p_{sb}^* > \overline{x}_1 + \underline{x}_2$.}  
As in Case~1a, this places a positive mass segment of $R$ inside $\mathcal{Z}$, violating Lemma~\ref{lem:finitechar}. 
Hence this sub-case cannot occur. 

\underline{Case 2b: $p_{sb}^* \leq \overline{x}_1 + \underline{x}_2$.}  

On the top edge 
\begin{equation*}
T := [\underline{x}_1,\overline{x}_1]\times\{\overline{x}_2\},
\end{equation*}
the measure $\mu$ has line density $(\overline{x}_2+k)f(x)$, which is strictly positive. 
Thus $\mu(A)>0$ for all measurable $A\subseteq T$.

Suppose further that $p_{sb}^* > \underline{x}_1+\overline{x}_2$. 
Then the segment $A = [\underline{x}_1, p_{sb}^*-\overline{x}_2] \times \{\overline{x}_2\} \subseteq T$ lies inside $\mathcal{Z}$ and $\mu(A) > 0$. 
Consider
\begin{equation*}
u_\delta(x) := \exp\!\left(\frac{x_2-\overline{x}_2}{\delta}\right).
\end{equation*}
This function is convex and non-decreasing. 
As $\delta\to 0$, $\int u_\delta\,d\mu|_{\mathcal{Z}}$ converges to $\mu(A)$, which is strictly positive—contradicting Lemma~\ref{lem:finitechar}. 
Therefore $p_{sb}^* \le \underline{x}_1+\overline{x}_2$.

Combining both bounds gives
\begin{equation*}
p_{sb}^*\le\min\{\overline{x}_1+\underline{x}_2,\;\underline{x}_1+\overline{x}_2\},
\end{equation*}
completing the proof.

\subsubsection{Proof of \Cref{theo:sbsuff}}\label{appendix:sbsuff}

Condition~(ii) of Proposition~\ref{prop:sbnec} ensures that the bottom edge $B$ carries negative linear density under~$\mu$. 
Along the left edge $L := \{\underline{x}_1\} \times [\underline{x}_2,\overline{x}_2]$, the linear density equals $-\underline{x}_1 f(x)$, which is nonpositive. 
The MM condition further guarantees that the interior of $X$ carries nonpositive mass. 
Together with $p_{sb}^* \leq \min \{\underline{x}_1 + \overline{x}_2,\;\overline{x}_1 + \underline{x}_2 \}$ from Condition~(ii), these properties imply that the only positive mass in $\mathcal{Z}$ lies at its coordinate-wise lowest point $(\underline{x}_1,\underline{x}_2)$. 
Hence $\int u\,d\mu|_{\mathcal{Z}} \leq 0$ for all convex, nondecreasing~$u$, establishing the sufficiency requirement from \Cref{lem:finitechar} on~$\mathcal{Z}$.

Next, fix any upper set $U \subseteq \mathcal{Y}$ and define
\begin{align*}
x_{1U} &:= \inf \{x_1 : (x_1,\overline{x}_2) \in U\},\\
x_{2U} &:= \inf \{x_2 : (\overline{x}_1,x_2) \in U\},
\end{align*}
and let the associated upper-right orthant be $O_U := [x_{1U},\overline{x}_1] \times [x_{2U},\overline{x}_2]$. 
By construction, $U \subseteq O_U \cap \mathcal{Y}$. 
Enlarging $U$ to $O_U \cap \mathcal{Y}$ only adds regions on which $\mu$ has nonpositive mass—by MM for the interior and by the negativity of the bottom and left edge densities—since no top or right edge is added in the process. 
Therefore, $\mu(U) \geq \mu(O_U \cap \mathcal{Y})$. 
Condition~(iii) then implies $\mu(U) \geq 0$.

Finally, Condition~(i) of Proposition~\ref{prop:sbnec} ensures that $\mu(\mathcal{Y}) = 0$. 
As in the layer-cake construction used in the proof of \Cref{theo:good-only}, this normalization allows the positive and negative parts of $\mu|_{\mathcal{Y}}$ to be compared directly. 
We then invoke the upper-set comparison theorem (Theorem~3.3.4 in \cite{muller2002comparison}) and conclude that
\begin{equation*}
\int u\,d\mu|_{\mathcal{Y}} \leq 0
\end{equation*}
for all coordinate-wise nonincreasing convex~$u$.
Thus, the sufficiency condition from \Cref{lem:finitechar} also holds on~$\mathcal{Y}$.

\subsubsection{Verification in \Cref{ex:sb}}\label{excal:sb}

To check Condition \textup{(iii)} of \Cref{theo:sbsuff}, consider three representative cases:
\begin{enumerate}
\item[(a)] \emph{Interior points.}  
For $x\in \operatorname{int}(X)$,
\begin{align*}
\mu([x_1,2]\times[x_2,0]\cap \mathcal{Y}) 
&\;\geq\; \mu([x_1,2]\times[x_2,0]) 
= 3(2-x_1)x_2 - 2x_2 + 1.5(2-x_1)\\
&=\;-3x_1x_2 - 1.5x_1 + 4x_2 + 3\;\geq\; 0.
\end{align*}
The first inequality holds because extending to the full rectangle adds only nonpositive mass (by MM).

\item[(b)] \emph{Left boundary above threshold.}  
If $x_1=1$ and $x_2\ge p_{sb}^*-1$, then
\begin{equation*}
\mu([1,2]\times[x_2,0]\cap \mathcal{Y}) 
= \mu([1,2]\times[x_2,0])
= 2x_2 + 1.5 \;\geq\; 0.
\end{equation*}

\item[(c)] \emph{Left boundary below threshold.}  
If $x_1=1$ and $-1<x_2<p_{sb}^*-1$, then
\begin{equation*}
\mu([1,2]\times[x_2,0]\cap \mathcal{Y}) 
= 2x_2+1.5 + 1.5(-0.54-x_2)^2 + (-0.54-x_2)\;\geq\;0.
\end{equation*}
\end{enumerate}
The case $x=(1,-1)$ is trivial, and other boundary cases with $x_2=-1$ are analogous.

Hence Condition \textup{(iii)} of \Cref{theo:sbsuff} is satisfied in all cases, and the Single-Bundle Posted Price mechanism is optimal in this example.

\subsection{Proofs for Ad-Tiered Posted Price}

\subsubsection{Verification in \Cref{ex:ad-tier}}\label{excal:ad-tier}

Substituting $p_{sb}^* = 0.80$, $M(x_1) = 1.5x_1^2 - 4.9x_1 + 3.6$.
Therefore, 
\begin{equation*}
\int_{t}^{p_{g}^*} M(x_1) = - 0.5t^3 + 2.45t^2 - 3.6t + 1.66  \geq 0 \quad \text{for all }t \in [1, p_{g}^*],
\end{equation*}
verifying Condition \textup{(ii)}.

In verifying Condition \textup{(iii)} we observe that,
\begin{equation*}
\mu([x_1,\overline{x}_1]\!\times\![\underline{x}_2,x_2]\!\cap\!\mathcal{W}) 
= 0.5(2-x_1) + 2(x_2+1) - 3(2-x_1)(x_2+1) \ge 0 
\quad \text{for all } x \in \mathcal{W}.
\end{equation*}

When verifying Condition \textup{(iv)} four cases arise,\\
\noindent Case 1. For $x \in \mathcal{Y}$,
\begin{equation*}
\mu([x_1,\overline{x}_1]\!\times\![x_2,\overline{x}_2]\!\cap\!\mathcal{Y}) 
= 0.5(2-x_1) - 2x_2 + 3(2-x_1)x_2 \ge 0.
\end{equation*}
\noindent Case 2. For $x_1 \in (1,p_{g}^*]$ and $x_2\!\in\![p_{sb}^*\!-\!p_{g}^*,\,p_{sb}^*\!-\!x_1]$,
\begin{equation*}
\mu([x_1,\overline{x}_1]\!\times\![x_2,\overline{x}_2]\!\cap\!\mathcal{Y}) 
= 0.5(2-x_1) - 2x_2 + 3(2-x_1)x_2 + 1.5(p_{sb}^* - x_1 - x_2)^2 \ge 0.
\end{equation*}
\noindent Case 3. For $x_1 = 1$ and $x_2\!\in\![p_{sb}^*\!-\!p_{g}^*,\,p_{sb}^*\!-\!1]$,
\begin{equation*}
\mu([x_1,\overline{x}_1]\!\times\![x_2,\overline{x}_2]\!\cap\!\mathcal{Y}) 
= 0.5 - 2x_2 + 3x_2 + 1.5(p_{sb}^* - 1 - x_2)^2 + p_{sb}^* - 1\ge 0.
\end{equation*}
\noindent Case 4. For $x_1 = 1$ and $x_2\!\in\![\,p_{sb}^*\!-\!1, 0]$,
\begin{equation*}
\mu([x_1,\overline{x}_1]\!\times\![x_2,\overline{x}_2]\!\cap\!\mathcal{Y}) 
= 0.5 - 2x_2 + 3x_2 + x_2\ge 0.
\end{equation*}


\subsection{Proofs for Uniform Distribution}

\subsubsection{Good-Only Posted Price}\label{app:unigood}

The resulting $\mu$-masses are shown in Figure \ref{fig:good-only-uniform}. 
Letting $c=(\overline{x}_1-\underline{x}_1)(\overline{x}_2-\underline{x}_2)$, the top and right edges have positive linear densities $\frac{1}{c}(\overline{x}_2+k)$ and $\frac{1}{c}\overline{x}_1$, respectively, while the left and bottom edges have negative linear densities $-\frac{1}{c}\underline{x}_1$ and $-\frac{1}{c}(\underline{x}_2+k)$, respectively. 
There is a point mass of $1$ at $(\underline{x}_1,\underline{x}_2)$ and a uniform surface density of $-\frac{3}{c}$ throughout the domain.

\begin{figure}[!hbt]
\centering
\includegraphics[width=4in]{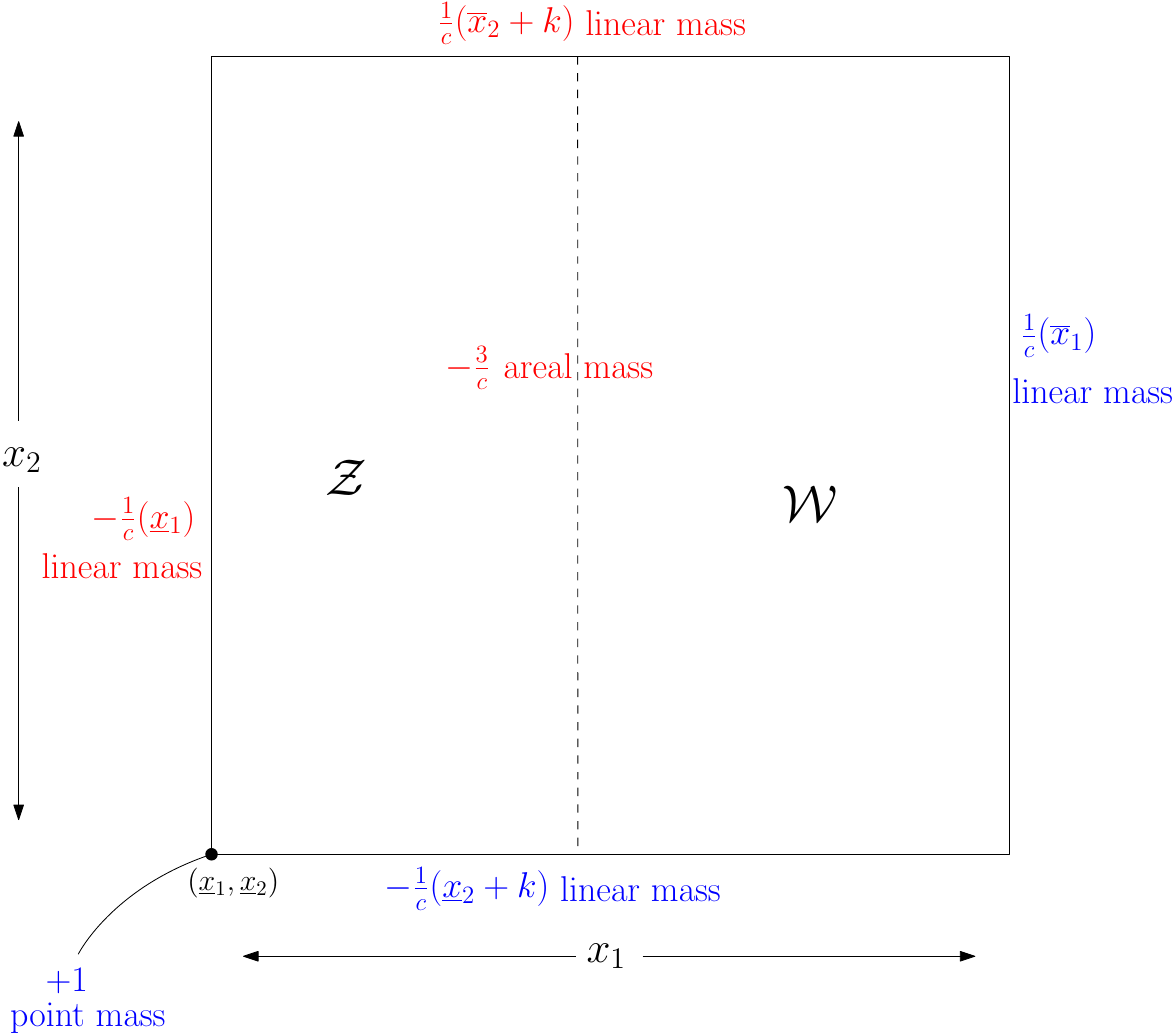}
\caption{Good-Only Posted Price : Uniform Distribution}
\label{fig:good-only-uniform}
\end{figure}

The forward implication follows directly from \Cref{prop:good-only}. 
For the converse, observe that the revenue-maximizing Good-Only Posted Price is $p^*_g=\max\{\frac{\overline{x}_1}{2},\underline{x}_1\}$. 
Given Condition (ii) of \Cref{theo:good-only}, we consider two cases and show that the remaining conditions are satisfied as well.

\noindent \underline{Case 1}: $\frac{\overline{x}_1}{2} > \underline{x}_1$.
Then $p^*_g = \frac{\overline{x}_1}{2}$ and $\mu(\mathcal{W}) = \frac{1}{c}\big[\overline{x}_1(\overline{x}_2-\underline{x}_2) + (\overline{x}_2 + k)(\frac{\overline{x}_1}{2}) - (\underline{x}_2 + k)\frac{(\overline{x}_2}{2}) - 3(\frac{\overline{x}_1}{2})(\overline{x}_2-\underline{x}_2)\big] = 0$.
Since $\mu(X) = 0$, we also have $\mu(\mathcal{Z}) = 0$, hence Condition (i) is satisfied.
We have $M(x_1) = 1 + \frac{1}{c}\big[(\overline{x}_2 + k)(x_1 - \underline{x}_1) - (\underline{x}_2 + k)(x_1 - \underline{x}_1) - \underline{x}_1(\overline{x}_2 - \underline{x}_2) - 3(\overline{x}_2 - \underline{x}_2)(x_1-\underline{x}_1)\big] = \frac{1}{\overline{x}_1 - \underline{x}_1}(\overline{x}_1 - 2x_1)$.
Therefore, $\int_{t}^{\frac{\overline{x}_1}{2}} M(x_1)dx_1 = \frac{1}{\overline{x}_1 - \underline{x}_1}\big[t^2 - t\overline{x}_1 + \frac{\overline{x}_1^2}{4} \big] = \frac{1}{\overline{x}_1-\underline{x}_1}\left(t-\frac{\overline{x}_1}{2}\right)^2 \geq 0$ for all $t \leq \frac{\overline{x}_1}{2}$, showing Condition (iii) is satisfied.

We now check Condition (iv) in four sub-cases depending on where $x$ lies.
When $x \in int(\mathcal{W})$, we have 
\begin{equation}\label{eq:mu-good-only}
\mu([x_1,\overline{x}_1] \times [\underline{x}_2,x_2]) = \frac{1}{c} \big[\overline{x}_1 (x_2 - \underline{x}_2) - (\underline{x}_2 + k)(\overline{x}_1 - x_1) - 3(x_2 - \underline{x}_2)(\overline{x}_1 - x_1) \big].
\end{equation}
Since the top edge has negative $\mu$-mass, the sum of $\mu$-masses in $\mathcal{W}$ other than the top edge must be positive.
Therefore, $\frac{1}{c}\Big[\overline{x}_1(\overline{x}_2 - \underline{x}_2) - (\underline{x}_2 + k)(\overline{x}_1 - \frac{\overline{x}_1}{2}) - 3(\overline{x}_2 - \underline{x}_2)(\overline{x}_1 - \frac{\overline{x}_1}{2})\Big] \geq 0$.
Substituting this into Equation (\ref{eq:mu-good-only}) and simplifying we get,
\begin{equation*}
\mu([x_1,\overline{x}_1] \times [\underline{x}_2,x_2]) \geq \frac{1}{c} \Bigg[\overline{x}_1(x_2 - \underline{x}_2)\Big[1 - \frac{\overline{x}_1 - x_1}{\overline{x}_1 - \frac{\overline{x}_1}{2}}\Big] - (\underline{x}_2 + k)(\overline{x}_1-x_1)\Big[1 - \frac{x_2 - \underline{x}_2}{\overline{x}_2 - \underline{x}_2}\Big] \Bigg] \geq 0,
\end{equation*}
where we use $k \leq |\overline{x}_2| < |\underline{x}_2|$ implies $-(\underline{x}_2 + k) > 0$ for the final inequality.

When $x$ is on the right or bottom edges, checking Condition (iv) is straightforward.
The sub-case when $x$ is on top-edge, we have
\begin{align}\label{eq:mu-top-edge}
\mu([x_1, \overline{x}_1] \times [\underline{x}_2, \overline{x}_2]) &= \frac{1}{c} \Big[\overline{x}_1 (\overline{x}_2 - \underline{x}_2) - (\underline{x}_2 + k)(\overline{x}_1 - x_1) + (\overline{x}_2 + k)(\overline{x}_1 - x_1) - 3(\overline{x}_1 - x_1)(\overline{x}_2 - \underline{x}_2) \Big] \notag\\
&= \frac{1}{c}\Big[(2x_1 - \overline{x}_1)(\overline{x}_2 - \underline{x}_2) \Big]\\
&\geq 0 \notag,
\end{align}
where the inequality is true since $x_1 > \frac{\overline{x}_1}{2}$ for $x \in \mathcal{W}$.

\noindent \underline{Case 2}: $\frac{\overline{x}_1}{2} \leq \underline{x}_1$. 
Then $p^*_g = \underline{x}_1$ and $\mu(\mathcal{W}) = \mu(X) = 0$ while $\mathcal{Z}$ is a null set.
So Condition (i) is satisfied, while Condition (iii) is trivially satisfied.

We check Condition (iv) in five sub-cases depending on where $x$ lies.
When $x \in int(\mathcal{W})$, \Cref{eq:mu-good-only} describes the orthant condition.
The top edge has negative $\mu$-mass, and the sum of $\mu$-masses on the left edge and the point $(\underline{x}_1, \underline{x}_2)$ is $1 - \frac{\underline{x}_1(\overline{x}_2 - \underline{x}_2)}{c} \leq 0$, where the inequality follows from the condition of this case.
Therefore, the sum of $\mu$-masses of $int(X)$, right and bottom edges is positive as $\mu(X) = 0$.
Hence, $\frac{1}{c}\Big[\overline{x}_1(\overline{x}_2 - \underline{x}_2) - (\underline{x}_2 + k)(\overline{x}_1 - \underline{x}_1) - 3(\overline{x}_2 - \underline{x}_2)(\overline{x}_1 - \underline{x}_1)\Big] \geq 0$.
Substituting this into \Cref{eq:mu-good-only} and simplifying we get,
\begin{equation*}
\mu([x_1,\overline{x}_1] \times [\underline{x}_2,x_2]) \geq \frac{1}{c} \Bigg[\overline{x}_1(x_2 - \underline{x}_2)\Big[1 - \frac{\overline{x}_1 - x_1}{\overline{x}_1 - \underline{x}_1}\Big] - (\underline{x}_2 + k)(\overline{x}_1-x_1)\Big[1 - \frac{x_2 - \underline{x}_2}{\overline{x}_2 - \underline{x}_2}\Big] \Bigg] \geq 0,
\end{equation*}

When $x$ is on the right or bottom edges, checking Condition (iv) is again straightforward.
When $x$ is on the top edge, the orthant condition is given by \Cref{eq:mu-top-edge}, and we note that $(2x_1 - \overline{x}_1)(\overline{x}_2 - \underline{x}_2) \geq 0$ as $2x_1 \geq 2\underline{x}_1 \geq \overline{x}_1$ in this case.

Finally, when $x$ is on the left edge, we have
\begin{align*}
\mu([\underline{x}_1, \overline{x}_1] \times [\underline{x}_2, x_2]) &= 1 + \frac{1}{c}\Big[\overline{x}_1(x_2 - \underline{x}_2) - (\underline{x}_2 + k)(\overline{x}_1 - \underline{x}_1) - \underline{x}_1(x_2 - \underline{x}_2) - 3(\overline{x}_1 - \underline{x}_1)(x_2 - \underline{x}_2) \Big]\\
&= \frac{-2x_2 - k + \overline{x}_2}{\overline{x}_2 - \underline{x}_2}\\
& \geq 0,
\end{align*}
where the inequality is true because $k \leq |\overline{x}_2| \leq |x_2|$.

\subsubsection{Single-Bundle Posted Price}\label{app:unisingle}
The resulting $\mu$-masses are shown in Figure \ref{fig:single-bundle-uniform}. 
The top and right edges have positive linear densities $\frac{k}{c^2}$ and $\frac{a+c}{c^2}$, respectively, while the left and bottom edges have negative linear densities $-\frac{a}{c^2}$ and $-\frac{k-c}{c^2}$, respectively. 
There is a point mass of $1$ at $(a,-c)$ and a uniform surface density of $-\frac{3}{c^2}$ throughout the domain.

\begin{figure}[!hbt]
\centering
\includegraphics[width=4in]{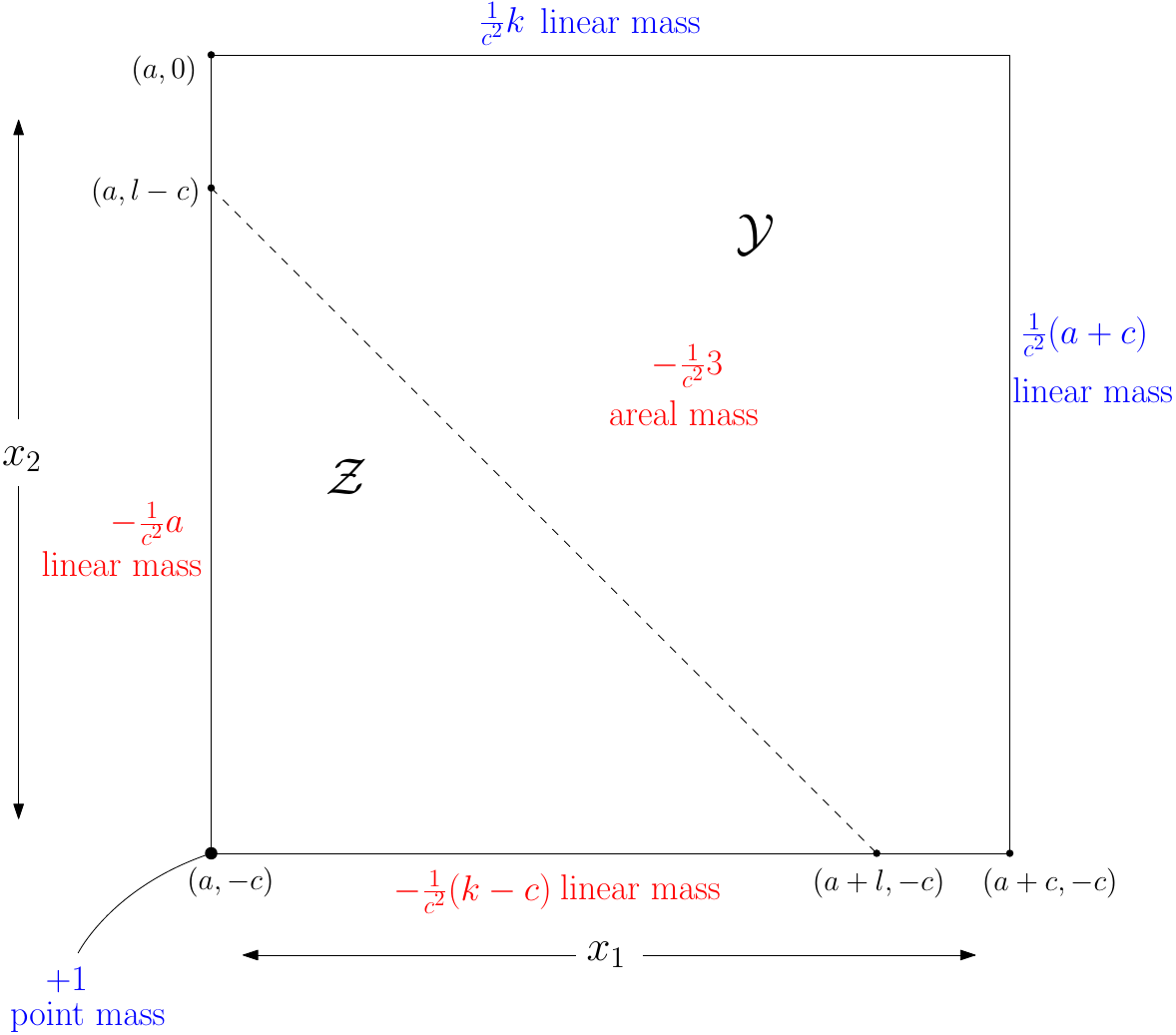}
\caption{Single-Bundle Posted Price : Uniform Distribution}
\label{fig:single-bundle-uniform}
\end{figure}

Let $\mathcal{Z}:=\{x\in X:x_1+x_2\leq a+l-c\}$, where we define $l$ by setting $\mu(\mathcal{Z})=0$. That is,
\begin{equation}\label{eq:ell}
l:=\frac{-(a+k-c)+\sqrt{(a+k-c)^2+6c^2}}{3}.
\end{equation}
Using $k\geq c$ and $a>0$, it is easy to verify that $l\in(0,c)$. 
We will show that the Single-Bundle Posted Price mechanism with price $p^*_{sb}:=a+l-c$ is optimal. 
Denote $\mathcal{Y}:=X\setminus\mathcal{Z}$. 
By definition of $l$, we have $\mu(\mathcal{Z})=0$, and since $\mu(X)=0$, it follows that $\mu(\mathcal{Y})=0$. 
Also, $k\geq c$ and $p^*_{sb}\leq a$. 
Therefore, conditions (i) and (ii) of \Cref{prop:sbnec} are satisfied. 
We now verify the orthant condition for $\mathcal{Y}$ from \Cref{theo:sbsuff} by considering eight cases.

\noindent \underline{Case 1:} $x$ lies on the top or right edge of $X$. 
The orthant condition is straightforward to verify.

\noindent \underline{Case 2:} $x \in X: x_1=a,x_2\geq l-c$. 
Then $\mu([a,a+c]\times[x_2,0]\cap\mathcal{Y})=\frac{1}{c^2}\big[(a+c)(-x_2)-a(-x_2)+kc-3(-x_2)c\big]=\frac{1}{c}(k+2x_2)$. 
Hence, it is enough to evaluate $E(k):=k+2(l-c)$ as $x_2\geq l-c$. 
Substituting $l$ from \Cref{eq:ell}, we get $E(k)=k-\frac{2}{3}(a+k-c)+\frac{2}{3}\sqrt{(a+k-c)^2+6c^2}-2c$. 
Differentiating with respect to $k$, we get $E’(k)=\frac{1}{3}+\frac{2}{3}\frac{a+k-c}{\sqrt{(a+k-c)^2+6c^2}}>0$, where the inequality is true because $k\geq c$. 
Hence, it is enough to check if $E(c)=\frac{2\sqrt{a^2+6c^2}-2a-3c}{3} \geq 0$.
But $E(c)$ decreasing in $a$, since its derivative is $\frac{2}{3}\left(\frac{a}{\sqrt{a^2+6c^2}}-1\right)<0$. 
Hence, using $a\leq c$, we have $E(c)\geq\frac{c}{3}(2\sqrt{7}-5)>0$. 
Therefore, the orthant condition holds in this case for all $k\geq c$.

\noindent \underline{Case 3:} $x \in X: x_1\geq a+l,x_2=-c$. 
Then $\mu([x_1,a+c]\times[-c,0]\cap\mathcal{Y})=\frac{1}{c^2}\big[(a+c)c-(k-c)(a+c-x_1)+k(a+c-x_1)-3(a+c-x_1)c\big]=\frac{1}{c}(2x_1-a-c)$. 
Hence, it is enough to check if the expression is nonnegative at $x_1=a+l$. 
Substituting $l$ using \Cref{eq:ell}, it is enough to check if $a-c-2k+2\sqrt{(a+k-c)^2+6c^2}\geq0$. 
Differentiating with respect to $k$, we get $2(-1+\frac{a+k-c}{\sqrt{(a+k-c)^2+6c^2}})<0$. 
Therefore, by setting $k = 3c$, it is enough to check if $a-7c+2\sqrt{(a+2c)^2+6c^2}\geq0$. 
This expression is increasing in $a$, so it is enough to evaluate it at $a=c/3$, where it equals $\frac{2c}{3}(\sqrt{103}-10)>0$. 
Hence, the orthant condition holds for all $a\geq c/3$ and $k\in[c,3c]$.

\noindent \underline{Case 4:} $x \in X: x_1+x_2=a+l-c$. 
Then $\mu([x_1,a+c]\times[x_2,0]\cap\mathcal{Y})=\frac{1}{c^2}\big[(a+c)(-x_2)+k(a+c-x_1)-3(-x_2)(a+c-x_1)\big]$. 
Substituting $x_2=a+l-c-x_1$, we need to ensure $E(x_1,k):=3x_1^2+(-5a+c-k-3l)x_1+(a+c)(2a-2c+k+2l)$ is positive for all $x_1\in[a,a+l]$ and $k \in [c,3c]$. 
Since $E$ is quadratic with positive leading coefficient, it is enough to show that its discriminant is negative, i.e., $12(a+c)(2a-2c+k+2l)-(-5a+c-k-3l)^2>0$ for all $k\in[c,3c]$. 
Eliminating $l$ using \Cref{eq:ell}, we get $-(a+k-c)^2+4(a+c)(a+k-c)-4a^2-18c^2+8c\sqrt{(a+k-c)^2+6c^2}$. 
We evaluate the derivative with respect to $k$ as $2(a+3c-k)+\frac{8c(a+k-c)}{\sqrt{(a+k-c)^2+6c^2}}$. 
The first term is positive since $k\leq3c$ and $a>0$, while the second term is positive since $k\geq c$ and $a>0$. 
Therefore, it is enough to show that the expression is positive at $k=c$. 
Substituting $k=c$, we get $-a^2+4ac-18c^2+8c\sqrt{a^2+6c^2}>-a^2+4ac+(8\sqrt{6}-18)c^2>0$, where we use $a\leq c$ to conclude the last inequality.

\noindent \underline{Case 5:} $x \in X: x_1 \in (a,a+c), x_2 \in (\max{-c,a+l-c-x_1},0)$. Then $\mu([x_1,a+c]\times[x_2,0]\cap\mathcal{Y})=\frac{1}{c^2}\Big[(a+c)(-x_2)+k(c-x_1)-3(-x_2)(c-x_1)\Big]$. To ensure positivity of the expression over the domain of this case, we argue that it suffices to consider the boundary. Indeed, the expression is bilinear in $x_1,x_2$ and hence has no interior minimum. The boundary consists of the edges covered in the previous four cases, and hence the expression is positive.

\noindent \underline{Case 6:} $x \in X: x_1=a,x_2\in[-c,l-c]$. 
Then we need to show that $\mu([a,a+c]\times[x_2,0]\cap\mathcal{Y})=\frac{1}{c^2}\big[-a(c-l)+kc+(a+c)(c-l+y)-3c(c-l+y)+\frac{3}{2}y^2\big]\geq0$ where $y=l-c-x_2$. 
After simplifying, we see that it is enough to show $E(y,k):=\frac{3}{2}y^2+(a-2c)y+(2cl+kc-2c^2)\geq0$ for all $y\in[0,l]$ and $k\in[\frac{9}{8}c,3c]$. 
Since the coefficient of $y^2$ is positive, it is enough to show that $(a-2c)^2-6(2cl+kc-2c^2)\leq0$. 
Eliminating $l$ using \Cref{eq:ell}, we conclude that it is enough to show $12c^2k^2+(4a^2c+32ac^2+16c^3)k-(a^4+8a^2c^2+32ac^3+32c^4)\geq0$. 
Since the derivative of the left-hand side with respect to $k$ is positive for all $k\geq0$, it is enough to evaluate it at $k=\frac{9}{8}c$. Substituting, we get $\frac{c^4}{16}(-16z^4-56z^2+64z+19)$ where $z:=\frac{a}{c}\in[1/3,1]$ since $a\in[c/3,c]$. 
The expression in $z$ is concave, since its second derivative is $-192z^2-112<0$, so its minimum on $[1/3,1]$ is attained at one of the endpoints. 
At $z=1/3$ and $z=1$, its values are $34$ and $11$, respectively. 
Therefore, the expression is positive for all $k\in[\frac{9}{8}c,3c]$, and hence $\mu([a,a+c]\times[x_2,0]\cap\mathcal{Y})\geq0$, as desired.

\noindent \underline{Case 7:} $x \in X : x_1\in[a,a+l],x_2=-c$. 
Then we need to show that $\mu([x_1,a+l]\times[-c,0]\cap\mathcal{Y})=\frac{1}{c^2}\big[(a+c)c+(c-l+y)k-(k-c)(c-l)-3(c-l+y)c+\frac{3}{2}y^2\big]\geq0$ where $y=a+l-x_1$. 
After simplifying, we see that it is enough to show $E(y,k):=\frac{3}{2}y^2+(k-3c)y+ac-c^2+2cl\geq0$ for all $y\in[0,l]$ and $k\in[c,3c]$. 
Eliminating $l$ using \Cref{eq:ell}, the discriminant condition ensuring $E(y,k)\geq0$ becomes $4c\sqrt{(a+k-c)^2+6c^2}\geq k^2-2ac-2ck+11c^2$. 
Since the right-hand side is positive for $a\in[c/3,c]$ and $k\in[c,3c]$, squaring and simplifying, it is enough to show $-k^4-9c^4+4ck^3+12kc^3-10k^2c^2+4ack^2+12ac^3+12a^2c^2+24ac^2k\geq0$.
The coefficients of $a$ and $a^2$ in this expression are positive, hence, using $a\geq c/3$, we substitute $a=c/3$. 
We then obtain $-3k^4+12ck^3-26c^2k^2+60c^3k-11c^4\geq0$. 
Letting $s:=k/c\in[1,3]$, this becomes $-3s^4+12s^3-26s^2+60s-11\geq0$. 
The expression is concave in $s$, since its second derivative is $-36s^2+72s-52<0$. 
Hence, its minimum on $[1,3]$ is attained at one of the endpoints, at which its values are $32$ and $16$, respectively. 
Therefore, the expression is positive for all $k\in[c,3c]$, and hence $\mu([x_1,a+l]\times[-c,0]\cap\mathcal{Y})\geq0$, as desired.

\noindent  \underline{Case 8:} $x \in X: x_1+x_2\leq a+l-c~\&~x_1 \neq a, x_2 \neq -c$. Then we need to show that $\mu([x_1,a+c]\times[x_2,0]\cap\mathcal{Y})=\frac{1}{c^2}\big[k(a+c-x_1)+(a+c)(-x_2)-3(a+c-x_1)(-x_2)+\frac{3}{2}(a+l-c-x_1-x_2)^2\big]\geq0$. This is a strictly convex expression, since its Hessian with respect to $(x_1,x_2)$ is $\frac{3}{c^2}I$. 
Its unique critical point is $(x_1^*,x_2^*)=(a-c+l+\frac{k}{3},\frac{a}{3}-\frac{5c}{3}+l)$, and hence it is enough to show that $E^*:=-\frac{1}{6}a^2+\frac{5}{3}ac-\frac{25}{6}c^2-al+5cl-\frac{1}{6}k^2-kl-\frac{3}{2}l^2+2ck\geq0$.

Using \Cref{eq:ell}, and letting $t:=\frac{a}{c}$ and $s:=\frac{k}{c}$, we obtain $\frac{E^*}{c^2}=-\frac{1}{6}s^2+\frac{2}{3}s-\frac{1}{6}t^2+\frac{1}{3}t-\frac{23}{6}+\frac{4}{3}\sqrt{(t+s-1)^2+6}$. 
This expression is increasing in $s$, since its derivative with respect to $s$ is $-\frac{s}{3}+\frac{2}{3}+\frac{4(t+s-1)}{3\sqrt{(t+s-1)^2+6}}>0$. Indeed, noting $t \in [1/3,1]$, for $s\in[1,2]$ this is immediate; while for $s\in[2,3]$, we have $t+s-1\in [\frac{4}{3},3]$, and hence $\frac{4(t+s-1)}{3\sqrt{(t+s-1)^2+6}}\geq\frac{16}{3\sqrt{135}}>1/3$. 
Hence, it is enough to evaluate $E^*$ at $s=1$. We then get $\frac{E^*}{c^2}=-\frac{1}{6}t^2+\frac{1}{3}t-\frac{10}{3}+\frac{4}{3}\sqrt{t^2+6}$. 
Since $\sqrt{t^2+6}\geq\sqrt{6}$, it is enough to show $-\frac{1}{6}t^2+\frac{1}{3}t-\frac{10}{3}+\frac{4}{3}\sqrt{6}>0$ for $t\in[1/3,1]$. The expression is concave, so its minimum is attained at one of the endpoints $t=1/3$ or $t=1$, at which its values are $-\frac{175}{54}+\frac{4\sqrt{6}}{3}>0$ and $-\frac{19}{6}+\frac{4\sqrt{6}}{3}>0$, respectively. 
Therefore, $E^*>0$ for all $a\in[c/3,c]$ and $k\in[c,3c]$, and hence $\mu([x_1,a+c]\times[x_2,0]\cap\mathcal{Y})\geq0$, as desired.

\subsubsection{Ad-Tiered Posted Price}\label{app:uniadtier}
The resulting $\mu$-masses are shown in Figure \ref{fig:ad-tier-uniform}. 
The top, right, and bottom edges have positive linear densities $\frac{k}{c^2}$, $\frac{a+c}{c^2}$, and $\frac{c-k}{c^2}$, respectively, while the left edge has negative linear densities $-\frac{a}{c^2}$. 
There is a point mass of $1$ at $(a,-c)$ and a uniform surface density of $-\frac{3}{c^2}$ throughout the domain.

\begin{figure}[!hbt]
\centering
\includegraphics[width=4in]{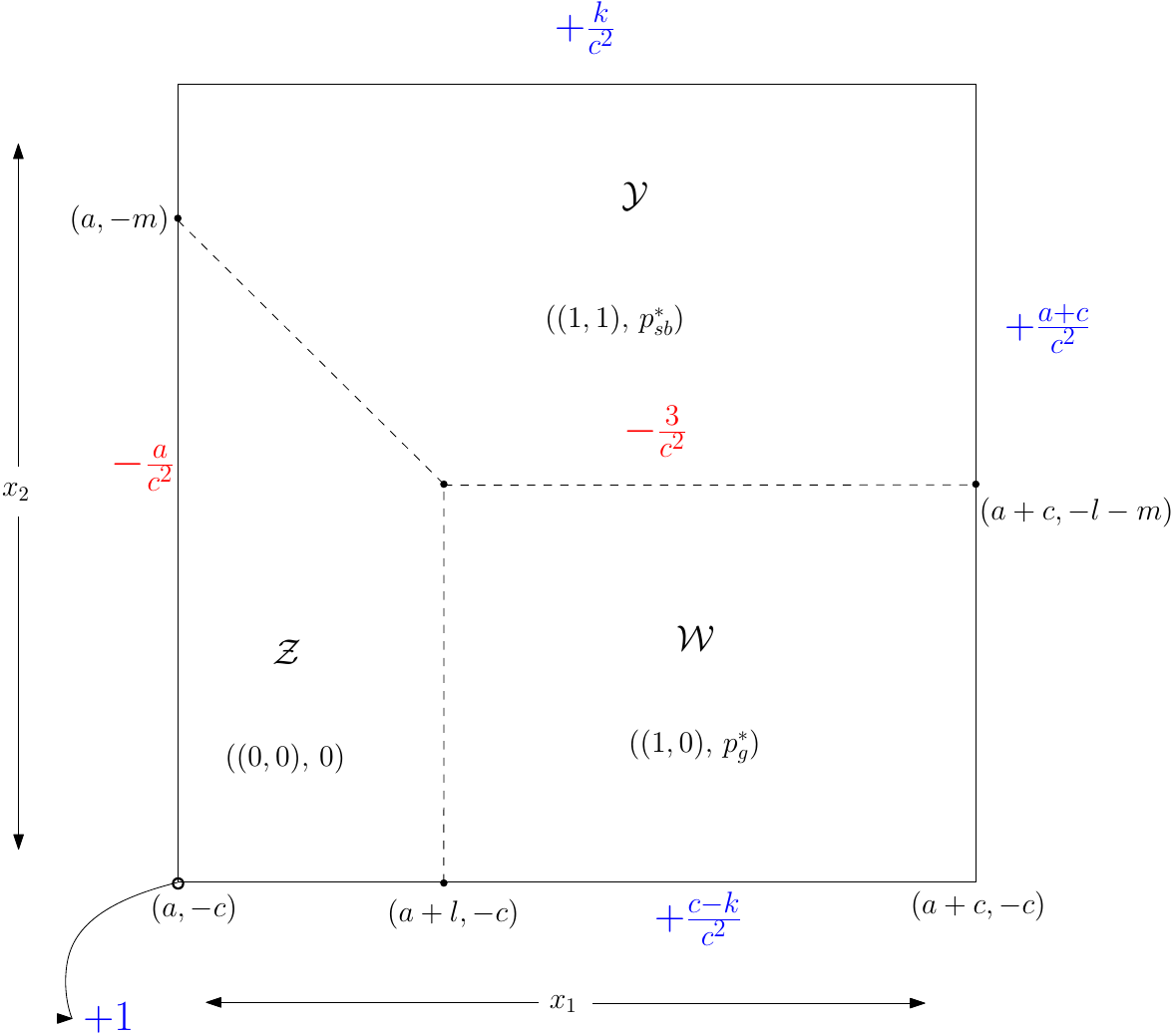}
\caption{Ad-Tiered Posted Price : Uniform Distribution}
\label{fig:ad-tier-uniform}
\end{figure}

We will show that Ad-Tiered Posted Price mechanism with prices $p^*_g = a + l$ and $p^*_{sb} = a - m$ is optimal, where $l$ and $m$ simultaneously solve
\begin{gather}
(a+c)(l+m) + kc - am - 3(l+m)c + \frac{3}{2}l^2 = 0 \label{eq:yzero}\\
(a+c)(c-l-m) + (c-k)(c-l) - 3(c-l-m)(c-l) = 0 \label{eq:wzero}
\end{gather}

Equations (\ref{eq:yzero}) and (\ref{eq:wzero}) are obtained by setting $\mu(\mathcal{Y}) = 0$ and $\mu(\mathcal{Z}) = 0$ (after normalizing), respectively.
We need to show that there exist feasible $l,m$ such that sufficient conditions (ii), (iii), and (iv) of \Cref{theo:ad-tier} are satisfied.
Feasibility requires $l,m \geq 0$ and $l+m \leq c$.

From \Cref{eq:yzero} we get 
\begin{equation}\label{eq:emm}
m = \frac{k}{2} + \frac{a-2c}{2c}l + \frac{3l^2}{4c}.
\end{equation}
Substituting this into \Cref{eq:wzero}, we obtain $P(l) := 9l^3 + 3(3a-2c)l^2 + 2(a^2-2ac+kc-4c^2)l + 2c(2c^2-2ac+ak) = 0$.
We observe that $P(0) = 2c\big(2c(c-a) + ak\big) > 0$ as $a \leq c$, and also $P(\frac{2c-a}{3}) = \frac{4c}{3}(a+c)(k-c) \leq 0$ as $k \leq c$.
So there exists at least one solution $l \in [0,\frac{2c-a}{3}]$ for $P(l)=0$.
\Cref{eq:yzero} implies $l + m = \frac{k}{2} + \frac{a}{2c}l + \frac{3l^2}{4c}$ , we get $\frac{d}{dl}(l+m) > 0$, so $l+m$ is strictly increasing in $l$.
Therefore, substituting $l = \frac{2c-a}{3}$ in RHS we get, $l + m \leq \frac{k}{2} + \frac{a}{2c}\frac{2c-a}{3} + \frac{3}{4c}(\frac{2c-a}{3})^2 = \frac{k}{2} + \frac{4c^2-a^2}{12c} \leq \frac{13}{16}c < c$, where the second inequality is true because $a \geq c/2$ and $k \leq c$.
Completing square in \Cref{eq:emm}, we get $m = \frac{k}{2} + \frac{3}{4c}(l - \frac{2c-a}{3})^2 - \frac{(2c-a)^2}{12c}$.
Since $k \geq c/2$, $a \geq c/2$ and $l \in [0, \frac{2c-a}{3}]$ we conclude $m \geq \frac{c}{16} > 0$.
Hence, we have a feasible pair $l$ and $m$.

For Condition (ii) of \Cref{theo:ad-tier}, we have (after normalizing) $M(x_1) = c^2 - a(c-m) + (c-k-3(c-l-m))(x_1-a) - \frac{3}{2}(x_1 - a)(2l - x_1 + a)$. 
We note $M(a+l) = 0$ as $\mu(\mathcal{Z}) = 0$ and $M'(a+l) = -2c - k + 3(l+m) \leq -\frac{c}{16} < 0$, where the first inequality is from $k \geq c/2$, and earlier established $l + m \leq \frac{13}{16}c$.
As $M(x_1)$ is strictly convex quadratic, we conclude $M(x_1) \geq 0$ for all $x_1 \in [a,a+l]$ which in turn implies Condition (ii).
Condition (iii)'s proof is similar to Case 1 in \Cref{lem:uniform-good-only}'s proof, hence we skip it here.

Finally, for Condition (iv), after normalizing densities, we consider cases.

\noindent \underline{Case 1.} $x \in X : x_1 = a, x_2 \in [-l-m, -m]$.
Then we need to show that $\mu([a,a+c] \times [x_2,0] \cap \mathcal{Y}) = (a+c)(-x_2) + kc -am - 3(-x_2)c + \frac{3}{2}(-m - x_2)^2 \geq 0$.
Replacing $-x_2$ with $y$ and simplifying we want to show $F(y) := (a-2c)y - am + kc + \frac{3}{2}(y-m)^2 \geq 0$ for all $y \in [m,l+m]$.
Then we have $F'(y) = a - 2c + 3(y-m) \leq a- 2c + 3l \leq 0$, where in the first inequality we use $y \leq l+m$, and the second $l \leq \frac{2c-a}{3}$.
So we evaluate $F(l+m) = (a-2c)(l+m) + kc - am + \frac{3}{2}l^2 = 0$, where in the final equation we directly use \Cref{eq:yzero}.

\noindent \underline{Case 2.} $x \in X : x_1 \in [a,a+l], x_2 \in [-l-m,-m]~\&~x_1 + x_2 \leq a-m$.
We need to show $\mu([x_1, a+l]\times[x_2,0] \cap \mathcal{Y}) = (a+c)(-x_2) + k(a+c-x_1) - 3(-x_2)(a+c-x_1) + \frac{3}{2}(a-m-x_1-x_2)^2 \geq 0$.
Substituting $s = x_1-a$ and $y = -x_2 - m$ we need to show $G(s,y) = (a+c)(m+y) + k(c-s) - 3(m+y)(c-s) + \frac{3}{2}(y-s)^2 \geq 0$ for all $0 \leq s \leq y \leq l$.
For a fixed $s$ we get $\frac{\partial G}{\partial y} = a - 2c + 3y \leq a - 2c + 3l \leq 0$ where in the second inequality we use $l \leq \frac{2c-a}{3}$.
So it is enough to show $G(s,l) \geq$ for all $s \in [0,l]$.
Substituting \Cref{eq:yzero} we get $G(s,l) = \frac{3}{2}s^2 + (3m-k)s + am$.
This is a quadratic with positive leading coefficient, so it is enough to show $(k-3m)^2 - 6am \leq 0$.
Using \Cref{eq:emm}, we conclude $k-3m = -\frac{k}{2} + \frac{3(2c-a)}{2c}l - \frac{9l^2}{4c} \leq -\frac{k}{2} + \frac{(2c-a)^2}{4c} \leq \frac{5c}{16}$, where in the first inequality we use the fact the quadratic $\frac{3(2c-a)}{2c}l - \frac{9l^2}{4c}$ attains maximum value of $\frac{(2c-a)^2}{4c}$ at $l = \frac{2c-a}{3}$, and in the second we use $k \geq c/2$ and $a \geq c/2$.
Therefore, we have $(k-3m)^2 \leq \frac{25c^2}{256}$.
Since $a \geq c/2$ and $m \geq c/16$ (argued in feasibility check), we get $6am \geq \frac{3c^2}{16}$.
Therefore, $6am \geq (k-3m)^2$ as required.

\noindent \underline{Case 3.} $x$ lies on the top or right edge of $X$. 
The orthant condition is straightforward to verify.

\noindent \underline{Case 4.} $x \in X : x_1 = a, x_2 \in [-m,0]$.
Then we need to show $\mu([a,a+l] \times [x_2,0] \cap \mathcal{Y}) = (a+c)(-x_2) + k(c) - a(-x_2) - 3(c)(-x_2) \geq 0$.
But this minimized at $x_2 = -m$ which we already covered in Case 1.

\noindent \underline{Case 5.} $x \in X : x_1 \in [a+l,a+c], x_2 = -l-m$.
Then we need to show $\mu([x_1, a+l]\times[-l-m,0] \cap \mathcal{Y}) = (a+c)(l+m) + k(a+c-x_1) - 3(a+c-x_1)(l+m) \geq 0$.
The minimum occurs at $x_1 = a+l$ which we already covered in Case 2.

\noindent \underline{Case 6.} $x \in X : x_1 + x_2 = a-m~\&~x_1 \in [a,a+l]$ is also covered in Case 2.

\noindent \underline{Case 7.} $x \in int(\mathcal{Y})$.
The orthant condition is bilinear in $x_1, x_2$ and so its minimum occurs on one of the edges which we covered in one of the last four cases.

As the remaining cases of $x \in X$ are redundant, we conclude the Condition (iv) of \Cref{theo:ad-tier} holds and hence the Ad-Tier Posted Price mechanism is optimal.


\end{document}